\documentclass[11pt,a4paper,thmsb,ukenglish,fleqn]{article}
\usepackage{epsf,  amsmath, amssymb, graphicx}
\usepackage{amsthm}
\usepackage[sort]{natbib}
\usepackage{a4wide}
\usepackage{subfigure}
\usepackage{xcolor}
\usepackage{pdflscape}
\usepackage{hyperref}

\usepackage{algpseudocode}

\usepackage{fancyhdr}
\newcommand{\mycomment}[1]{}

\newcommand{\R}{\ensuremath{\Bbb{R}}}
\newcommand{\N}{\ensuremath{\Bbb{N}}}
\newcommand{\E}{\ensuremath{\Bbb{E}}}
\newcommand{\Var}{\ensuremath{\mathrm{Var}}}
\newcommand{\Cor}{\ensuremath{\mathrm{Cor}}}
\newcommand{\Cov}{\ensuremath{\mathrm{Cov}}}
\newcommand{\diag}{\ensuremath{\mathrm{diag}}}

\newcommand{\leb}{\ensuremath{\mathrm{Leb}}}
\def\Levy{L\'{e}vy }

\newtheorem{theorem}{Theorem}

\newtheorem{algorithm}[theorem]{Algorithm}

\newtheorem{corollary}{Corollary}

\newtheorem{definition}[theorem]{Definition}
\newtheorem{example}{Example}

\newtheorem{proposition}{Proposition}
\newtheorem{remark}{Remark}

\allowdisplaybreaks
\begin{document}
\title{Modelling, simulation and inference for multivariate time series of counts}
\author{\textsc{Almut E.~D.~Veraart} \\
\textit{Department of Mathematics, Imperial College London}\\
\textit{ 180 Queen's Gate, 
 London, SW7 2AZ, 
UK}  \\
\texttt{a.veraart@imperial.ac.uk}
}
\maketitle

\begin{abstract}
This article presents a new continuous-time modelling framework  
for multivariate time series of counts which have an   infinitely divisible marginal distribution. The model is based on a  mixed moving average process driven by \Levy noise -- called a trawl process -- 
where the serial  correlation and the cross-sectional dependence are modelled independently of each other. Such processes can exhibit  short or long memory.
We derive a stochastic simulation algorithm  and a   statistical inference method for such processes. The new methodology is then applied to  high frequency financial data, where we investigate the relationship between the number of limit order submissions and deletions in a limit order book.

\end{abstract}

 \noindent{\bf Keywords:} 
 Count data, continuous time modelling of multivariate time series, trawl processes, infinitely divisible, Poisson mixtures,  multivariate negative binomial law, limit order book\\ 
\noindent {\bf Mathematics Subject Classification:}  60G10,  60G55, 	60E07, 62M10, 62P05
\maketitle{}

\section{Introduction}
Time series of counts can be viewed as realisations of non-negative integer-valued stochastic processes and arise in various applications in the natural,   life and  social sciences. As such there has been very active research in the various fields and 
recent textbooks treatments can be found in 
 \cite{CameronTrivedi1998,KedemFokianos2002, Winkelmann2003, DHLR2015} and we refer to \cite{Davisetal1999,McKenzie2003, Ferlandetal2006,Weiss2008,CuiLund2009,DavisWu2009,JungTremayne2011} for recent surveys and some new developments of the literature.

However, most of these previous works focus on univariate time series of counts and the literature on multivariate extensions is rather sparse and almost exclusively deals with models formulated in discrete time and borrow ideas from traditional autoregressive time series models. E.g.~\cite{FrankeRao1995}  and \cite{Latour1997} introduced the first-order integer-valued autoregression model, which is based on the generalised Steutel and van Harn (1979) thinning  operator.  Recently, \cite{BoudreaultCharpentier2011} applied such models to  earthquake counts.
Also, the recent handbook on discrete-valued time series by \cite{DHLR2015} contains the   chapter by  \cite{Karlis2015} who surveys recent developments in multivariate count time series models. 

One challenge in handling multivariate time series is the  modelling of  the cross-sectional dependence. While for continuous distributions the theory of copulas presents a powerful toolbox, it has been pointed out by \cite{GenestNeslehova2007} that a problem arises in the discrete context due to the non-uniqueness of the associated copula. This can be addressed by using the continuous extension approach by \cite{DenuitLambert2005}. Indeed, for instance \cite{HeinenRengifo2007} introduce a multivariate time series model for counts based on copulas applied to continuously extended discrete random variables and fit the model to the numbers of trades of various assets at the New York stock exchange. Also, \cite{KLL2015} study discrete copula distributions with time-varying marginals and dependence structure in financial econometrics. 
Motivated by the reliability literature, \cite{LindskogMcNeil2003}  introduced the so-called common Poisson shock model to describe the arrival of  insurance claims in multiple locations or losses due to credit defaults of various types of counterparty.

While the models mentioned above are interesting in their own right, the goal of this article is more ambitious since it formulates a more  general modelling framework which can handle a variety of marginal distributions as well as different types of serial dependence including, in particular, both short and long memory specifications. That said, motivated by an application in financial econometrics and recognising the success the class of \Levy processes has in such settings, we focus exclusively on models whose marginal distribution is infinitely divisible. This assumption puts  a restriction on the cross-sectional dependence due to the well-known result by \cite{Feller1968}, which says that a random vector with infinitely divisible distribution on $\mathbb{N}^n$ always has  non-negatively correlated components. Moreover, any non-degenerate distribution on $\mathbb{N}^n$ is infinitely divisible  if and only if it can be expressed as a discrete  compound Poisson distribution. We will see that this is nevertheless a very rich class of distributions and suitable for our application to high frequency financial data.

The new modelling framework  is based on so-called multivariate integer-valued trawl processes, which are special cases of multivariate mixed moving average processes where the driving noise is given by an integer-valued \Levy basis. 

In the univariate case, trawl processes -- not necessarily restricted to the integer-valued case -- have been introduced  by  \cite{BN2011}. Also,   \cite{NVG2015} used such processes in an hierarchical model in the context of extreme value theory. The univariate integer-valued case has been developed in detail in  \cite{BNLSV2014}.  \cite{ShephardYang2016a} studied likelihood inference  for a particular subclass of an integer-valued trawl process and, more recently, \cite{ShephardYang2016b} used such processes to build an econometric model for fleeting discrete price moves.
While the multivariate extension was already briefly mentioned in  \cite{BNLSV2014}, this article develops the theory of multivariate integer-valued trawl (MVIT) processes  in detail and presents new methodology for stochastic simulation and statistical inference for such processes and applies the new results to high frequency financial data from a limit order book.
The key feature of MIVT processes, which makes them powerful for a wide range of applications is the fact that the serial dependence and the marginal distribution can be modelled independently of each other, which is  for instance not the case in the famous DARMA models, see \cite{JacobsLewis1978a,JacobsLewis1978b}. As such we will present parsimonious ways of parameterising the serial correlation and will show that we can accommodate both short and long memory processes as well as seasonal fluctuations.
Moreover, since MITV processes are formulated in continuous time, we can handle both asynchronous and not necessarily equally spaced observations, which is particularly important in a multivariate set-up.

The motivation for this study comes from high frequency financial econometrics where discrete data arise in a variety of scenarios, e.g.~high frequent price moves for stocks with fixed tick size resemble step functions supported on a fixed grid. Also, the number of trades can give us an indication of market activity and is widely analysed in the industry. In this article, we will apply  our new methodology  to model the relationship between the number of submitted and deleted limit orders in a limit order book, which are key quantities in high frequency trading.

The outline of this article is as follows. 
Section \ref{Section:Model} introduces the class of multivariate integer-valued trawl processes and presents its probabilistic properties.
Section \ref{Section:crosssection} gives a detailed overview of parametric model specifications focusing on a variety of different cases for modelling the serial correlation. Moreover, we present relevant examples of multivariate  marginal distributions which fall into the infinitely divisible framework. In particular, as pointed out by  \cite{NKarlis2008}, the negative binomial distribution often appears to be a suitable candidate for various applications. Hence we will derive several approaches to defining a multivariate infinitely divisible distribution which allows for univariate negative binomial marginal law.
In Section \ref{Section:SimInf} we will derive an algorithm to simulate from MIVT processes and develop a statistical inference methodology which we will also test in a simulation study.
Section \ref{Section:Empirics} applies the new methodology to limit order book data.
Finally, Section \ref{Section:conclusion} concludes. 
The proofs of the theoretical results are relegated to the Appendix, Section \ref{proofs}, and Section \ref{sim} provides  more details on the  algorithms used in the simulation study.

\section{Multivariate integer-valued trawl processes}\label{Section:Model}
\subsection{Integer-valued \Levy bases as driving noise}
Throughout the paper, we denote by $(\Omega, \mathcal{F}, (\mathcal{F}_t), P)$ the underlying filtered probability space satisfying the usual conditions.
Also, we choose a set  $E \subset \R^d$ ($d\in \N$)  and let the corresponding Borel $\sigma$-algebra  be denoted by  $\mathcal{E}=\mathcal{B}(E)$. 
Next we define a Radon measure  $\mu$  on $(E, \mathcal{E})$,  which by definition satisfies $\mu(B)< \infty$ for every compact measurable set $B \in \mathcal{E}$.

In the following, we will always assume that the Assumption (A1) stated below holds.
\begin{description}
\item[Assumption (A1)] Let  $E=\R^n\times [0,1]\times \R$ for $n\in \N$ and let  $N$ be a homogeneous Poisson random measure on $E$  with intensity measure
$\mu(d{\bf y}, dx, ds) =\E(N (d{\bf y}, dx, ds)) = \nu(d{\bf y}) dx ds$,
where $\nu$ is a \Levy measure concentrated on  $\mathbb{Z}^n\setminus \{{\bf 0}\}$ and satisfying\\
$\int_{\mathbb{R}^n} \min(1, ||{\bf y}||) \nu(d{\bf y}) < \infty$.
\end{description}
Using the Poisson random measure, we can  define an integer-valued \Levy basis as follows.

\begin{definition}\label{DefIVLevyBasis}
 Suppose that   $N$ is a homogeneous Poisson random measure on $(E, \mathcal{E})$ satisfying Assumption (A1).
An \emph{$\mathbb{Z}^n$-valued, homogeneous \Levy basis} on $([0,1]\times\mathbb{R}, \mathcal{B}([0,1]\times\mathbb{R}))$ is defined as
\begin{align}\label{IVL}
{\bf L}(dx, ds) =(L^{(1)}(dx, ds), \dots, L^{(n)}(dx, ds))^{\top}= \int_{\R^n} {\bf y} N(d{\bf y}, dx, ds).
\end{align}
\end{definition}

From the definition, we can immediately see that ${\bf L}$ is infinitely divisible with characteristic function given by 
\begin{align*}
\mathbb{E}(\exp(i \boldsymbol{ \theta}^{\top} {\bf L}(dx, ds)))=\exp(\text{C}_{{\bf L}(dx,ds)}(\boldsymbol{ \theta})), \quad \boldsymbol{ \theta} \in \R^n.
\end{align*}
Here,  $\text{C}$ denotes the associated 
cumulant function, which is the (distinguished) logarithm of the characteristic function. It can we written as   
\begin{align*}
\text{C}_{{\bf L}(dx, ds)}(\boldsymbol{ \theta}) = \text{C}_{{\bf L}'}(\boldsymbol{ \theta} )dx ds,
\end{align*}
where the random vector ${\bf L}'$ denotes the corresponding \emph{\Levy seed} with cumulant function given by 
\begin{align}\label{Mult-LK}
\text{C}_{{\bf L}'}(\boldsymbol{ \theta} ) =  \int_{\mathbb{R}^n}\left(e^{i\boldsymbol{\theta}^{\top} {\bf y}}-1\right) \nu(d{\bf y}),
\end{align}
where $\nu$ denotes the corresponding \Levy measure defined above.

\begin{remark}
It is important to note that the \Levy seed specifies the  homogeneous \Levy basis uniquely, and vice versa, with any homogeneous \Levy basis we can associate a unique \Levy seed.
Hence, in modelling terms, it will later be sufficient to discuss various modelling choices for the corresponding \Levy seed, since this will fully characterise the associated \Levy basis.  
\end{remark}
\begin{remark}
Based on the \Levy seed, we can define a \Levy process denoted by
$({\bf L}'_t)_{t\geq 0}$, when setting  ${\bf L}_1' = {\bf L}'$. Clearly, in this case, we get  $\text{C}_{{\bf L}'_t}(\boldsymbol{ \theta}) = t \text{C}_{{\bf L}'}(\boldsymbol{ \theta})$. 
\end{remark}

 Following the construction in \citet[Theorem 4.3]{Sato1999}, we model the \Levy seed by an $n$-dimensional compound Poisson random variable  given by 
  \begin{align*}
  {\bf L}'=\sum_{j=1}^{N_1}{\bf Z}_j,
  \end{align*}
  where $N=(N_t)_{t\geq 0}$ is an homogeneous Poisson process of rate $v>0$ and the $({\bf Z}_j)_{j\in \mathbb{N}}$ form a sequence of i.i.d.~random variables independent of $N$ and  which have no atom in ${\bf 0}$, i.e.~not all components are simultaneously equal to zero, more precisely, $\mathbb{P}({\bf Z}_j={\bf 0}) = 0$ for all $j$.

    \begin{remark}
  Recall that by modelling the \Levy seed by a multivariate compound Poisson process we can only allow for positive correlations between the components.
  \end{remark}

\subsection{The trawls}
Following the approach presented in \cite{BN2011}, see also \cite{BNLSV2014}, we now define the so-called \emph{trawls}.

\begin{definition}\label{TrawlDef}
We call  a Borel set $A\subset  [0,1]\times (-\infty, 0]$ such that $\leb(A) < \infty$
a  \emph{trawl}. Further, we set
\begin{align}\label{trawl}
A_t = A + (0,t), \quad t \in \mathbb{R}.
\end{align}
\end{definition}
The above definition implies that the trawl at time $t$ is just the shifted trawl from time $0$.

\begin{remark}
Note that the size of the trawl does not change over time, i.e.~we have
$\leb(A_t)=\leb(A)$ for all $t$.
\end{remark}

Clearly, there is a wide class of sets which can be considered as trawls. Throughout the paper, we will hence narrow down our focus, and will concentrate  on a particular subclass of trawls which can be written as 
\begin{align}\label{StTrawl}
A = \{(x, s): s \leq 0,\ 0\leq x \leq  d(s)\},
\end{align}
where $ d:(-\infty, 0] \mapsto [0,1]$ is a continuous function such that $\leb(A) < \infty$. Typically we refer to $d$ as the \emph{trawl function}.
In such a semi-parametric setting, we can easily deduce that 
\begin{align}\label{dstar}
 \leb(A) =  \int_{-\infty}^0 d(s) ds.
\end{align}
Moreover, the corresponding trawl at time $t$ is given by 
\begin{align*}
A_t = A + (0,t)= \{(x, s): s \leq t,\ 0\leq x \leq  d(s-t)\}.
\end{align*}

\begin{definition}
Let $A$ denote a trawl  given by \eqref{StTrawl}.
If $d(0)=1$ and $d$ is monotonically non-decreasing, then we call $A$ a \emph{monotonic trawl}.
\end{definition}

\begin{example}
\label{Ex-Ex0}
Let   $d(s) = \exp(\lambda s)$ for $\lambda > 0, s \leq 0$. Then the corresponding trawl
is monotonic with 
$A_t = A + (0,t)= \{(x, s): s \leq t,\ 0\leq x \leq  \exp(\lambda(s-t))\}$.
\end{example}

In our multivariate framework, we will choose $n$ trawls denoted by  
 $A^{(i)}= A_0^{(i)}$. Then we set $A^{(i)}_t = A^{(i)} + (0,t)$ for $i \in \{1, \dots, n\}$.
When we work with trawls of the type \eqref{StTrawl}, we will denote by $d^{(i)}$ the corresponding trawl functions.

\subsection{The multivariate integer-valued trawl process and its properties}

\begin{definition}
The stationary multivariate integer-valued trawl (MIVT) process is defined by 
\begin{align*}
{\bf Y}_t = \left(L^{(1)}(A_t^{(1)}), \dots, L^{(n)}(A_t^{(n)})\right)^{\top},\quad t\in \mathbb{R},
\end{align*}
where  each component is given by 
\begin{align*}
Y^{(i)}_t=  L^{(i)}(A_t^{(i)})=\int_{[0,1]\times \R}\mathbf{I}_{A^{(i)}}(x,s-t)L^{(i)}(dx,ds), \quad i \in \{1, \dots, n\}, 
\end{align*}
where $\mathbf{I}$ denotes the indicator function.
\end{definition}
Since the trawls have finite Lebesgue measure, the integrals above are well-defined in the sense of \cite{RajRos89}. 

When we define $\mathbf{I}_{\bf{A}}(x,s-t)=\diag(\mathbf{I}_{A^{(1)}}(x,s-t), \dots,\mathbf{I}_{A^{(n)}}(x,s-t))$, then we can represent the MIVT process as 
\begin{align*}
{\bf Y}_t = \int_{\R^n\times [0,1]\times \R} {\bf y} \mathbf{I}_{\bf{A}}(x,s-t)N(d{\bf y}, dx,ds), \quad t \in \mathbb{R},
\end{align*}
which shows that we are dealing with a  special case of a multivariate mixed moving average process. 

The law of the MIVT process is fully characterised by its characteristic function, which we shall present next.
\begin{proposition}\label{CharFctTrawl}
For any ${\boldsymbol \theta} \in \R^n$, the characteristic function of ${\bf Y}_t$ is given by $\mathbb{E}(\exp(i{\boldsymbol \theta}^{\top}{\bf Y}_t))=\exp(C_{{\bf Y}_t}({\boldsymbol \theta}))$, where the corresponding cumulant function is given by 
\begin{multline*}
C_{{\bf Y}_t}({\boldsymbol \theta})=
\sum_{k=1}^n \sum_{\substack{1\leq i_1, \dots, i_k \leq n: \\ i_{\nu}\not 
= i_{\mu}, \text{ for } \nu \not = \mu}}
\leb\left(\bigcap_{l=1}^kA^{(i_l)}\setminus \bigcup_{\substack{1\leq j\leq n,\\ j\not \in\{i_1, \dots, i_k\}}} A^{(j)} \right) 
C_{(L^{(i_1)},\dots,L^{(i_k)})}((\theta_{i_1},\dots,\theta_{i_k})^{\top}).
\end{multline*}
\end{proposition}

\begin{corollary}
In the special case when $A^{(1)}=\cdots=A^{(n)}=A$, the characteristic function  simplifies to
$\mathbb{E}(\exp(i{\boldsymbol \theta}^{\top}{\bf Y}_t)) =
\exp\left(\leb(A) \text{C}_{{\bf L'}}({\boldsymbol \theta})\right)$. \end{corollary}
This is an important result, which implies that  to any infinitely divisible integer-valued law  $\pi$, say,  there exists a stationary integer-valued trawl process having $\pi$ as its  marginal law.

\subsubsection{Cross-sectional and serial dependence}
Let us now focus on the cross-sectional and the serial dependence of multivariate integer-valued trawl processes.

First, the cross-sectional dependence is  entirely characterised through the multivariate \Levy measure $\nu$.
For instance, when we focus on the pair of the $i$th and the $j$th component for $i,j \in \{1, \dots, n\}$,  we define the corresponding 
joint \Levy measure by 
\begin{align*}
\nu^{(i,j)}(d\cdot, d\cdot) = \int_{\R} \dots \int_{\R}
\nu(dy_1, \dots, dy_{i-1}, d\cdot,dy_{i+1}, \dots, dy_{j-1},d\cdot, dy_{j+1}, \dots, dy_n).
\end{align*}
Then the  covariance between the $i$th and the $j$th \Levy seed is given by 
\begin{align*}
\kappa_{i,j}:=\int_{\R}\int_{\R}
y_i y_j\nu^{(i,j)}(dy_i, dy_j).
\end{align*}
Relevant specifications of $\nu$ will be discussed in Section \ref{MultDist}.

Second, the serial dependence is determined through the trawls. More precisely, following \cite{BN2011}, we introduce the so-called  \emph{autocorrelator} between the $i$th and the $j$th component, which is defined as 
\begin{align*}
R_{ij}(h) = \leb(A_0^{(i)}\cap A_h^{(j)}), \quad h \geq 0.
\end{align*}

Let us now focus on the autocorrelators 
for trawls of  type \eqref{StTrawl}.
\begin{proposition}
Suppose the trawls $A^{(i)}$, $i \in \{1, \dots, n\}$ are of type  \eqref{StTrawl}.
Then for $h\geq 0$ the intersection of two trawls is given by 
\begin{align*}
A^{(i)}\cap A^{(j)}_h =\{(x,s): s \leq 0, 0 \leq x \leq \min\{d^{(i)}(s), d^{(j)}(s-h)\}\}.
\end{align*}
I.e.~the autocorrelator satisfies
\begin{align*}
R_{ij}(h)
&= \int_{-\infty}^0\min\{d^{(i)}(s), d^{(j)}(s-h)\} ds.
\end{align*}
\end{proposition}
The proof is straightforward and hence omitted.

\begin{remark}
Note that the autocorrelators can be computed as soon as the corresponding trawl functions and their parameters are known. We will come back to this aspect when we discuss inference for trawl processes in Section \ref{InfSection}.
\end{remark}
Let us consider a canonical example when the trawl functions are given by exponential functions.
\begin{example}\label{Ex-expR12}
Let 
$d^{(i)}(s) =\exp(\lambda_i s)$. For  $i, j \in \{1,\dots, n\}$ suppose that $\lambda_i < \lambda_j$. Then for $s \leq 0$ we have that $e^{\lambda_i s} \geq e^{\lambda_j s}$ and hence
$A^{(i)}\cap A^{(j)} = A^{(j)}$. Hence $\leb(A^{(i)}\cap A^{(j)})=\leb(A^{(j)})=1/\lambda_j$.
Similarly, we get that 
$R_{ij}(h)=\leb(A^{(i)}\cap A^{(j)}_h)=\frac{1}{\lambda_j}e^{-\lambda_j h}$, for $h\geq 0$.
\end{example}

For monotonic trawl functions we observe that there are two possible scenarios: Either, one trawl function is always \lq below\rq\ the other one, which implies that 
\begin{align*}R_{ij}(h) =\min(\leb(A^{(i)}),\leb(A^{(j)})),
\end{align*}
see e.g.~Example \ref{Ex-expR12},
 or the trawl functions intersect each other. 
  In the latter case, suppose there is one intersection of $d^{(i)}$ and $d^{(j)}$ at time  $s^* <0$, say.  Consider the scenario  when $d^{(i)}(s) \leq d^{(j)}(s)$  
for $s \leq s^*$ and $d^{(j)}(s) \leq d^{(i)}(s)$ for $s^* \leq s \leq 0$. Then
\begin{align*}
R_{ij}(0)&=\leb (A^{(i)}\cap A^{(j)}) 
= \int_{-\infty}^{s^*}d^{(i)}(s)ds + \int_{s^*}^0d^{(j)}(s)ds.
\end{align*}
Extensions to a multi-root scenario are straightforward.

Clearly, the autocorrelators are closely related to the autocorrelation function. More precisely, we have the following result, which follows directly from the expression of the cumulant function of the multivariate trawl process.

\begin{proposition}
The covariance between two (possibly shifted) components $1\leq i \leq j\leq n$ for $t, h \geq 0$ is given by 
\begin{align*}
\rho_{ij}(h) &=
\Cov\left(L^{(i)}(A_t^{(i)}), L^{(j)}(A_{t+h}^{(j)})\right)
=  \leb\left(A^{(i)}\cap A_h^{(j)}\right)\left(\int_{\R}\int_{\R}
y_i y_j\nu^{(i,j)}(dy_i, dy_j)\right)\\
&= R_{ij}(h) \kappa_{i,j}.
\end{align*}
Also, the corresponding 
 auto- and cross-correlation function
 is given by 
\begin{align*}
r_{ij}(h)&:=\Cor\left(L^{(i)}(A_t^{(i)}), L^{(j)}(A_{t+h}^{(j)})\right)
=
\frac{\leb(A^{(i)}\cap A_h^{(j)}) \left(\int_{\R}\int_{\R}
y_i y_j\nu^{(i,j)}(dy_i, dy_j)\right)}{\sqrt{\leb(A^{(i)}) \Var(L^{'(i)})   \leb(A^{(j)}) \Var(L^{'(j)})}}\\
&=\frac{R_{ij}(h)}{\sqrt{\leb(A^{(i)}) \leb(A^{(j)})}}
\frac{\kappa_{i,j}}{\sqrt{\Var(L^{'(i)})    \Var(L^{'(j)})}},
\end{align*}
i.e.~the autocorrelation function is \emph{proportional} to the autocorrelators.
\end{proposition}
We will come back to the above result when we turn our attention to parametric inference for MIVT processes in Section \ref{InfSection}.

\section{Parametric specifications}\label{Section:crosssection} 
In order to showcase the flexibility of the new modelling framework, we will discuss various parametric model specifications in this section, where we start off by considering specifications of the trawl, followed by models for the multivariate \Levy seed.
\subsection{Specifying the trawl function}
We have already covered the case of an exponential trawl function  above and will now present alternative choices for the trawl functions and their corresponding autocorrelators, see also \cite{BNLSV2014} for other examples.

While an exponential trawl leads to an exponentially decaying autocorrelation function, we sometimes need model specifications which exhibit a more slowly decaying autocorrelation function. Such trawl functions can be constructed from the exponential trawl function by randomising the memory parameter as we will describe in the following example. 

To simplify the notation we will in the following
supress the indices $i$ for the corresponding component in the multivariate construction, i.e.~we set $d=d^{(i)}$ and  do not write the sub-/superscripts for the corresponding parameters.  
\begin{example}\label{Ex-ExGIG}
 Define the trawl function by 
\begin{align*}
 d(z) = \int_{0}^{\infty} e^{\lambda z} \pi(d \lambda), \qquad \text{for} \; z \leq 0,
\end{align*}
for a probability measure $\pi$ on $(0, \infty)$.
Suppose that  $\pi$ is absolutely continuous with density $f_{\pi}$, then the corresponding trawl function can be written as
\begin{align*}
 d(z) = \int_{0}^{\infty} e^{\lambda z} f_{\pi}(\lambda) d\lambda,
\end{align*}
which again leads to a monotonic trawl function.
The corresponding autocorrelation function is given by
\begin{align*}
r(h) &= 
\frac{\int_0^{\infty}\frac{1}{\lambda}e^{-\lambda h} \pi(d \lambda)}{\int_0^{\infty}\frac{1}{\lambda} \pi(d \lambda)},
\end{align*}
 assuming that  $\int_0^{\infty}\frac{1}{\lambda} \pi(d \lambda)< \infty$.
\end{example}

\cite{BNLSV2014} discuss various constructions of that type depending on different choices of the probability measure $\pi$ and we refer to that article for more details on the computations. 

In applications, we often  assume that $\pi$ is absolutely continuous with respect to the Lebesgue measure and we denote its density by $f_{\pi}$.
A very flexible parametric framework can be obtained by choosing $f_{\pi}$ to be a generalised inverse Gaussian (GIG) density as we shall discuss in the next example.

\begin{example}\label{ExGIGDetails}
Suppose that  $f_{\pi}$ is the density of the GIG distribution, i.e.
\begin{align}\label{GIGdensity}
f_{\pi}(x) = \frac{(\gamma/\delta)^{\nu}}{2K_{\nu}(\delta \gamma)} x^{\nu -1} \exp\left(-\frac{1}{2}(\delta^2 x^{-1} + \gamma^2 x)\right),
\end{align}
where $\nu \in \mathbb{R}$ and $\gamma$ and $\delta$ are both nonnegative and not simultaneously equal to zero. Here we denote by  $K_{\nu}(\cdot)$ the modified Bessel function of the third kind.
Straightforward computation show that the corresponding 
 trawl function is given by
\begin{align*}
 d(z) &= \left(1-\frac{2z}{\gamma^2}\right)^{-\frac{\nu}{2}}
\frac{K_{\nu}(\delta\gamma \sqrt{1-\frac{2z}{\gamma^2}})}{K_{\nu}(\delta \gamma)},
\end{align*}
and the corresponding size of the trawl set equals 
\begin{align*}
\leb(A)&= \frac{(\gamma/\delta) K_{\nu-1}(\delta \gamma)}{K_{\nu}(\delta \gamma)}.
\end{align*}
Moreover,  the autocorrelation function is given by
\begin{align*}
r(h) &=  \frac{K_{\nu-1}(\delta \sqrt{\gamma^2+ 2h})}{K_{\nu-1}(\delta \gamma)} \left(1+\frac{2h}{\gamma^2} \right)^{\frac{1}{2}(1-\nu)}.
\end{align*}
\end{example}
Some special cases of the GIG distribution include the inverse Gaussian and the gamma distribution, which lead to interesting parametric examples which we shall study next.
\begin{example}\label{supIGTrawl}
Suppose we choose an 
 inverse Gaussian (IG) density function for $f_{\pi}$. Then we obtain the so-called
sup-IG trawl function, which can be written as
\begin{align*}
d(z) = \left(1-\frac{2z}{\gamma^2}\right)^{-1/2}\exp\left(\delta \gamma\left(1-\sqrt{1-\frac{2z}{\gamma^2}}\right)\right),
\end{align*}
for nonnegative parameters $\delta, \gamma$ which are assumed not to be simultaneously equal to zero.
Then we have that $\leb(A) = \frac{\gamma}{\delta}$
and the corresponding 
autocorrelation function is given by 
\begin{align*}
r(h) = \exp\left(\delta \gamma \left(1-\sqrt{1+\frac{2h}{\gamma^2}} \right) \right), \quad h \geq 0.
\end{align*}

\end{example}

Next, we consider an example where the trawl function decays according to a power law.
\begin{example}
A long memory specification can be obtained when the probability measure $\pi$ is chosen to have Gamma distribution. In that case, we obtain a trawl function given by 
\begin{align*}
d(z) = \left( 1-\frac{z}{\alpha}\right)^{-H}, \quad \alpha > 0, H>1.
\end{align*}
Then $\leb(A)=\alpha/(H-1)$.
Also,
\begin{align*}
r(h)= \left(1 + \frac{h}{\alpha}\right)^{1-H}.
\end{align*}
I.e.~when $H\in (1,2]$ we have a stationary long memory model, and when when $H>2$ we obtain a stationary short memory model.
\end{example}
Finally, we consider the case of a seasonal trawl function.
\begin{example}
A seasonally varying trawl function can be obtained by setting $d(z)=d_m(z)d_s(z)$, where $d_m$ is a monotonic trawl function and $d_s$ is a periodic seasonal function. E.g.~as discussed in \cite[Example 9]{BNLSV2014}, we can consider the following functional form
\begin{align*}
d(z) = \frac{1}{2}\exp(\lambda x)\left[\cos(az) +1\right], \quad \text{where } a =2\pi\psi.
\end{align*}
Here $\lambda >0$ determines how quickly the function decays, whereas $\psi\in \R$ denotes the period of the season. In this case, we obtain
$\leb(A)=(2\lambda^2 +a^2)/(2\lambda(\lambda^2+a^2))$ and 
\begin{align*}
r(h) = \frac{e^{-\lambda h}}{2\lambda(\lambda^2 + a^2)}\left(\lambda^2\cos(ah)-a\lambda \sin(ah) + \lambda^2 + a^2\right).
\end{align*}
Note that this construction leads to a seasonal autocorrelation function, but not to seasonality in the levels of the trawl process.
\end{example}

\subsection{Modelling the cross-sectional dependence}\label{MultDist}
The trawl process is completely specified, as soon as both the trawls and the marginal distribution of the multivariate \Levy seed are specified.
When it comes to infinitely divisible discrete distributions, the Poisson distribution is the natural starting point and we will review multivariate extensions in Section \ref{MultP}.
However, since many count data exhibit overdispersion, it is crucial that we go beyond the Poisson framework. 
In the univariate context, there have been a variety of articles on suitable discrete distributions, see e.g.~\cite{PuigValero2006} and 
 \cite{NKarlis2008} amongst others. However, the literature on parametric classes of multivariate infinitely divisible discrete distributions with support on $\mathbb{N}^n$ is rather sparse. 
We know that any such distribution necessarily is of discrete compound Poisson type, see 
\cite{Feller1968,OspinaGerber1987, Sundt2000}, and always has 
 non-negatively correlated components.
 In Section \ref{PoissonMix} we will discuss a possible parametrisation based on Poisson mixtures of random additive-effect-type models.

\subsubsection{Multivariate Poisson marginal distribution}\label{MultP} 
As before, we denote by ${\bf L'} =(L^{'(1)},\dots, L^{'(n)})^{\top}$ the \Levy seed. To start off with we present a multivariate Poisson law for the \Levy seed. 
In order to introduce dependence between the Poisson random variables, one typically uses a so-called \emph{common factor approach}, which we outline in the following, see e.g.~\cite{Karlis2002,KarlisMeligkotsidou2005}.

Suppose that we have $m\in \mathbb{N}$ independent random variables $X^{(i)} \sim Poi(\theta_i)$ for $i=1, \dots, m$, and set ${\bf X}=(X^{(1)},\dots, X^{(m)})^{\top}$.

Let ${\bf A}$ denote a $n\times m$-matrix (for $n \in \mathbb{N}$) with 0-1 entries and having no duplicate columns.
We then set 
${\bf L'} ={\bf A}{\bf X}$, which clearly follows a 
multivariate Poisson distribution. The corresponding mean and variance can be easily computed and are given by
$\mathbb{E}({\bf L'})= {\bf A} {\bf M}$ and $\Var({\bf L'}) = {\bf A} {\boldsymbol \Sigma} {\bf A}^{\top}$, respectively, 
where ${\bf M}=\mathbb{E}({\bf X})$ and ${\boldsymbol \Sigma} = \Var({\bf X})$. Since the components $X^{(i)}$ are independent, we have
${\boldsymbol \Sigma} = \diag(\theta_1, \dots, \theta_m)$ and ${\bf M}^{\top} = (\theta_1, \dots, \theta_m)$.
The above construction implies that
$L^{'(i)}\sim Poi(v_i)$, where $v_i = \sum_{k=1}^ma_{ik}\theta_i$.
Also, for $i\not = j$ we have that
\begin{align*}
\Cor(L^{'(i)}, L^{'(j)})= \frac{\sum_{k=1}^{m} a_{ik}\theta_k a_{kj}}{\sqrt{\sum_{k=1}^{m} a_{ik}^2\theta_k  \sum_{k=1}^{m} a_{jk}^2\theta_k }}.
\end{align*}

Let us study some relevant examples within this modelling framework.

\begin{example}
An $n$-dimensional model with one common factor between all components can be obtained by choosing $m=n+1$, and 
\begin{align*}
{\bf A} =  \left( \begin{array}{ccccc}1&0&\cdots&\cdots&1 \\
0 & 1& 0& \cdots & 1\\
\vdots & \ddots & \ddots & \ddots & \vdots\\
0 & \cdots & 0 & 1 & 1
\end{array}\right),
&& {\bf X} = \left(\begin{array}{c} X^{(1)}\\
\hdots\\
X^{(n)}\\
X^{(0)}
\end{array} \right)
\end{align*}
and independent Poisson random variables 
  $X^{(i)}\sim Poi(\theta_i)$, for  $i=0,1,\dots,n$. Then we have 
  \begin{align*}
  L^{'(1)} = X^{(1)} + X^{(0)},&&  L^{'(2)} = X^{(2)} + X^{(0)}, && \cdots, &&
  L^{'(n)} = X^{(n)} + X^{(0)}.
  \end{align*}
  Here each component has marginal Poisson distribution, i.e.~$L^{'(i)}\sim Poi(\theta_i + \theta_0)$ and for $i\not = j$ we have  that $\Cov(L^{'(i)},L^{'(j)}) = \theta_0$. 
   \end{example}
   Beyond the bivariate case, the example above presents a rather restrictive model for applications since it only allows for one common factor.  A less sparse choice of ${\bf A}$ would allow for more flexible model specifications.
   Let us consider a more realistic example in the trivariate case next.
   \begin{example}
Consider a model of the type
  \begin{align*}
 L^{'(1)} &= X^{(1)} + X^{(12)} + X^{(13)} + X^{(123)},\\
  L^{'(2)} &= X^{(2)} + X^{(12)} + X^{(23)} + X^{(123)},\\
   L^{'(3)} &= X^{(3)} + X^{(13)} + X^{(23)} + X^{(123)}
  \end{align*}
  for independent Poisson random variables $X^{(i)}$ with parameters $\theta_i$, for \\ $i \in \{\{1\}, \{2\}, \{3\}, \{12\}, \{13\}, \{23\}, \{123\} \}$.
  Such a model specification corresponds to the choice of 
  \begin{align*}
  {\bf A } = \left(
\begin{array}{ccccccc}
1 & 0 & 0 & 1 & 1& 0 & 1\\
0 & 1& 0 & 1& 0 & 1 & 1\\
0 & 0 & 1& 0 & 1 & 1& 1
\end{array} \right),
&& {\bf X}= \left(
X^{(1)}, X^{(2)}, X^{(3)},X^{(12)},X^{(13)},X^{(23)},X^{(123)}
 \right)^{\top}.
  \end{align*}
  Here we have that $L^{'(1)}\sim Poi(\theta_1+\theta_{12}+\theta_{13}+\theta_{123})$,
  $L^{'(2)}\sim Poi(\theta_2+\theta_{12}+\theta_{23}+\theta_{123})$ and
  $L^{'(3)}\sim Poi(\theta_3+\theta_{13}+\theta_{23}+\theta_{123})$.
  \end{example}
  The above example treats a very general case which allows for all possible bivariate as well as a trivariate covariation effect. A slightly simpler specification is given in the next example, which only considers pairwise interaction terms. 
   \begin{example}
    Choosing 
  \begin{align*}
  {\bf A } = \left(
\begin{array}{cccccc}
1 & 0 & 0 & 1 & 1& 0 \\
0 & 1& 0 & 1& 0 & 1 \\
0 & 0 & 1& 0 & 1 & 1
\end{array} \right),
&& {\bf X}= \left(
X^{(1)}, X^{(2)}, X^{(3)},X^{(12)},X^{(13)},X^{(23)}
 \right)^{\top},
  \end{align*}
   results in a trivariate model of the form 
  \begin{align*}
 L^{'(1)} = X^{(1)} + X^{(12)} + X^{(13)},&&
  L^{'(2)} = X^{(2)} + X^{(12)} + X^{(23)}, &&
   L^{'(3)} = X^{(3)} + X^{(13)} + X^{(23)},
  \end{align*}
  for independent Poisson random variables $X^{(i)}$ with parameters $\theta_i$, for\\ $i \in \{\{1\}, \{2\}, \{3\}, \{12\}, \{13\}, \{23\}\}$.
  Then we have that $L^{'(1)}\sim Poi(\theta_1+\theta_{12}+\theta_{13})$,
  $L^{'(2)}\sim Poi(\theta_2+\theta_{12}+\theta_{23})$ and
  $L^{'(3)}\sim Poi(\theta_3+\theta_{13}+\theta_{23})$; also, 
  \begin{align*}
  \Var({\bf L'}) = \left( 
\begin{array}{ccc} 
\theta_1 + \theta_{12}+\theta_{13} & \theta_{12} & \theta_{13}\\
 \theta_{12} & \theta_2 + \theta_{12} +\theta_{23}& \theta_{23}\\
 \theta_{13} & \theta_{23} & \theta_{3}+\theta_{13}+\theta_{23}
\end{array} \right).
  \end{align*}
  \end{example}

\subsubsection{Multivariate discrete compound Poisson  marginal distribution obtained from Poisson mixtures} \label{PoissonMix} 

While the Poisson distribution is a good starting point in the context of modelling count data, for many applications it might be too restrictive. In particular, often one needs to work with distributions which allow for overdispersion, i.e.~that the variance is bigger than the mean. 

Since we are interested in staying within the class of discrete infinitely divisible stochastic processes, the most general class of distributions we can consider are the discrete compound Poisson distributions.
 To this end, we model the \Levy seed by an $n$-dimensional compound Poisson random variable, see e.g.~\citet[Theorem 4.3]{Sato1999},  given by 
  \begin{align*}
  {\bf L}'=\sum_{j=1}^{N_1}{\bf C}_j,
  \end{align*}
  where $N=(N_t)_{t\geq 0}$ is an homogeneous Poisson process of rate $v>0$ and the $({\bf C}_j)_{j\in \mathbb{N}}$ form a sequence of i.i.d.~random variables independent of $N$ and  which have no atom in ${\bf 0}$, i.e.~not all components are simultaneously equal to zero, more precisely, $\mathbb{P}({\bf C}_j={\bf 0}) = 0$ for all $j$.\\ \

\noindent{\bf General Poisson mixtures}\\
\noindent
Previous research has clearly documented that Poisson mixture distributions provide a flexible class of distributions which are suitable for various applications, see e.g.~\cite{KarlisXekalaki2005} for a review.

In this section, we are going to introduce a parsimonious parametric model class for the  $n$-dimensional \Levy seed ${\bf L}'$, which uses Poisson mixtures and is
 based on the  results in Section 5 of \cite{BNBlaesildSeshadri1992}. 
To this end, consider random variables $X_1, \dots, X_n$ and $Z_1, \dots, Z_n$ for $n \in \mathbb{N}$ and assume that
conditionally on $\{Z_1, \dots, Z_n\}$ the $X_1, \dots, X_n$ are independent and Poisson distributed with means given by the $\{Z_1, \dots, Z_n\}$.

We then model the joint distribution of the $\{Z_1, \dots, Z_n\}$ by a so-called additive effect model as follows:
\begin{align*}
Z_i = \alpha_i U + V_i, \quad i = 1, \dots, n,
\end{align*}
where the random variables $U, V_1, \dots, V_n$ are independent and the $\alpha_1, \dots, \alpha_n$ are nonnegative parameters.

We can easily derive the probability generating function of the joint distribution of $X_1, \dots, X_n$, cf.~\citet[Section 5]{BNBlaesildSeshadri1992}:
\begin{align*}
\mathbb{E} (t_1^{X_1}\cdots t_n^{X_n}) = 
M_U\left( \sum_{i=1}^n \alpha_i(t_i-1)\right) \prod_{i=1}^n M_{V_i}(t_i-1),
\end{align*}
where we denote by $M_X(\theta)=\mathbb{E}(e^{\theta X})$ the moment generating function of a random variable $X$ with parameter $\theta$.

Also, we can  compute the means and the covariance function of the $Y_i$s and find that
\begin{align*}
\mathbb{E}(X_i) = \alpha_i \mathbb{E}(U) + \mathbb{E}(V_i), \quad i =1, \dots, n,
\end{align*}
and
\begin{align*}
\Cov(X_i,X_j) = \left \{ 
\begin{array}{ll}
\alpha_i^2 \Var(U) + \Var(V_i) + \alpha_i \mathbb{E}(U) + \mathbb{E}(V_i), & \text{ if } i= j,\\
\alpha_i \alpha_j \Var(U), & \text{ if } i \not = j.
\end{array} 
\right.
\end{align*}

Next we derive the  
 joint law of $(X_1, \dots,X_n)$, see \cite{BNBlaesildSeshadri1992} for the bivariate case.
\begin{proposition}\label{JointLaw}
In the additive random effect model the joint law of $(X_1, \dots,X_n)$ is given by 
\begin{multline*}
P(X_1=x_1, \dots, X_n=x_n) 
= \frac{1}{x_1!\cdots x_n!} \sum_{j_1=0}^{x_1}\cdots \sum_{j_n=0}^{x_n} {x_1 \choose j_1}\cdots{x_n \choose j_n}\alpha_1^{j_1}\cdots \alpha_j^{j_n}
\\
\cdot \E(U^{j_1+\cdots+j_n}e^{-(\alpha_1 + \cdots + \alpha_n)U})  \prod_{k=1}^n \E(V_k^{ y_k-j_k}e^{-V_k}).
\end{multline*}
\end{proposition}

Next, we establish the key result of this section, which links the Poisson mixture distribution based on an additive effect model to a discrete compound Poisson distribution. Recall, see e.g.~\citet[p.~18]{Sato1999}, that an $n$-dimensional compound Poisson random variable ${\bf L}' =\sum_{i=1}^{N_1}{\bf C}_i$ has Laplace transform given by
\begin{align}\label{LaplaceCPP1}
\mathcal{L}_{{\bf L}'}({\boldsymbol \theta})=\E(e^{-{\boldsymbol \theta}^{\top}{\bf L}'})=\exp(v (\mathcal{L}_{{\bf C}}({\boldsymbol \theta})-1)),
\end{align}
where $v>0$ is the intensity  of the  Poisson process $N$  and $\mathcal{L}_{{\bf C}}({\boldsymbol \theta})$ is the Laplace transform of the i.i.d.~jump sizes.

\begin{proposition}\label{CPRep}
The Poisson mixture model of random-additive-effect type can be represented as a discrete compound Poisson distribution with rate
\begin{align*}
v &= -\left( \overline K_U(\alpha) + \sum_{i=1}^n \overline K_{V_i}(1)\right) ,
\end{align*}
where $\alpha = \sum_{i=1}^n\alpha_i$ and $\overline K$ denotes the kumulant function, i.e.~the logarithm of the Laplace transform, 
and the jump size distribution has Laplace transform given by
\begin{align*}
\mathcal{L}_{{\bf C}}({\boldsymbol \theta}) &= \frac{1}{v}\left\{\sum_{k=1}^{\infty} \left(\sum_{i=1}^n \alpha_ie^{-\theta_i}\right)^k q_k^{(U)}
+\sum_{i=1}^n\sum_{k=1}^{\infty} e^{-\theta_i k} q_k^{(V_i)} \right\},
\end{align*}
where
\begin{align*}
q_k^{(U)} = \frac{1}{k!} \int_{\R}e^{-\alpha x}   x^k\nu_U(dx), &&
q_k^{(V_i)} = \int_{\R} \frac{x^k}{k!}e^{-x}\nu_{V_i}(dx), \quad \text{ for } i \in \{1, \dots, n\},
\end{align*}
where $\nu_U$ and  $\nu_{V_i}$ denotes the \Levy measure of $U$ and $V_i$, respectively.
\end{proposition}
The above result is very important since we need the compound Poisson representation to efficiently simulate the trawl process, as we shall discuss  in Section \ref{simalgo}.\\ \ 

\noindent{\bf Multivariate negative binomial distribution}\\
\noindent
In situations where the count data are overdispersed and call for distributions other than the Poisson one, we can in principle choose from a great variety of discrete compound Poisson distributions. 
Motivated by our empirical study, see Section \ref{Section:Empirics}, and also the results in \cite{BNLSV2014}, we investigate the case of a negative binomial marginal law in more detail since this is one of the infinitely divisible distributions which can cope with overdispersion.

Recall that we say that a random variable $X$ has negative binomial law with parameters $\kappa>0, 0<p<1$, i.e.~$X\sim NB(\kappa, p)$ if its probability mass function is given by 
\begin{align*}
\mathbb{P}(X=x) = {\kappa+x-1 \choose x} p^x(1-p)^{\kappa}, \quad x \in \{0,1,\dots\}.
\end{align*}
Its probability generating function is given by
$G(t) =\E(t^X) = \left(1-\frac{p}{1-p}(t-1) \right)^{-\kappa}$.
Also, recall that a random variable $X$ is said to be gamma distributed with parameters $a,b>0$, i.e.~$X\sim\Gamma(a,b)$ if its probability density is given by $f(x) =\frac{b^a}{\Gamma(a)}x^{a-1}e^{-bx}$, for $x>0$.

Now, we set $U \sim \Gamma(\kappa,1)$ and $V_i \sim \Gamma(\kappa_i,\beta_i^{-1})$ in the Poisson mixture model. Then the probability generating function of $(X_1,\dots,X_n)$ is given by 
\begin{align*}
\mathbb{E} (t_1^{X_1}\cdots t_n^{X_n}) 
&= \left(1+  \sum_{i=1}^n \alpha_i(t_i-1)\right)^{-\kappa} \prod_{i=1}^n (1-\beta_i(t_i-1))^{-\kappa_i}.
\end{align*}
Next we are going to describe three examples, see  \citet[Example 5.3]{BNBlaesildSeshadri1992}, which lead to negative binomial marginals.
The first example, Example \ref{ex_ind}, covers the case of independent components, in the second example, Example \ref{ex_dep}, the fully dependent case is achieved through the presence of a common factor, and the third example, Example \ref{ex_indanddep}, combines the previous two cases by allowing for both a common (dependent) factor and additional independent components.
\begin{example}[Independence case]\label{ex_ind}
We set $\alpha_i \equiv 0$, for $i =1, \dots, n$ and choose $V_i \sim \Gamma(\kappa_i,1/\beta_i)$.
Then
$\mathbb{E} (t_1^{X_1}\cdots t_n^{X_n}) 
=  \prod_{i=1}^n (1-\beta_i(t_i-1))^{-\kappa_i}$, which implies that the $X_i$ are independent and satisfy $X_i \sim NB(\kappa_i, \beta_i/(1+\beta_i))$.
\end{example}
\begin{example}[Dependence through common factor]\label{ex_dep}
Choose $U \sim \Gamma(\kappa, 1)$ and $V_i \equiv 0$, for  $i =1, \dots, n$. Note that such a construction extends the bivariate case considered in \cite{ArbousKerrich1951}. Then
$\mathbb{E} (t_1^{X_1}\cdots t_n^{X_n}) = 
 \left(1+  \sum_{i=1}^n \alpha_i(t_i-1)\right)^{-\kappa}$, which implies that 
$X_i \sim NB(\kappa, \alpha_i/(1+\alpha_i))$
and also
$\sum_{i=1}^n X_i \sim NB\left(\kappa, \frac{\sum_{i=1}^n \alpha_i}{1+\sum_{i=1}^n \alpha_i} \right)$.
\end{example}
\begin{example}[Dependence through common factor and additional independent factors]\label{ex_indanddep}
Suppose that $U \sim \Gamma(\kappa, 1)$ and $V_i \sim \Gamma(\kappa_i,1/\alpha_i)$. 
Then one can write
$Z_i=\alpha_i(U + W_i)$, for 
$U \sim \Gamma(\kappa, 1)$ and $W_i \sim \Gamma(\kappa_i,1)$.
Then we can deduce that 
$\mathbb{E} (t_1^{X_1}\cdots t_n^{X_n}) 
= \left(1+  \sum_{i=1}^n \alpha_i(t_i-1)\right)^{-\kappa} \prod_{i=1}^n (1-\alpha_i(t_i-1))^{-\kappa_i}$.
Hence
$X_i \sim NB(\kappa+\kappa_i, \alpha_i/(1+\alpha_i))$.
\end{example}

\begin{remark}
The dependence concepts used here can be considered as Poisson mixtures of the first kind, see \cite{KarlisXekalaki2005}.
\end{remark}

We conclude this section by deriving the compound Poisson representation of the multivariate negative binomial distribution.

\begin{example}\label{Ex} As before, let
$U \sim \Gamma(\kappa,1), V_i \sim \Gamma(\kappa_i,1/\beta_i)$. 
Recall that for $X\sim \Gamma(a,b)$, $\E(e^{i\theta X}) = (1-i\theta/b)^{-a}$.
Hence  $\mathcal{L}_U(\theta)=(1+\theta)^{-\kappa}$, and $\mathcal{L}_V(\theta)= (1+\theta\beta_i)^{-\kappa_i}$. Also, 
$\overline K_U(\theta)=-\kappa\log(1+\theta)$, and $\overline K_V(\theta)= -\kappa_i\log(1+\theta\beta_i)$.
Then the rate  in the compound Poisson representation is given by 
$v= \kappa\log(1+\alpha) + \sum_{i=1}^n \kappa_i\log(1+\beta_i)$.
Further, we have
$\nu_U(dx) = \kappa x^{-1}e^{- x}dx$, and $\nu_{V_i}(dx)=\kappa_ix^{-1}e^{-x/\beta_i}dx$.
Then we can compute
\begin{align*}
q_k^{(U)} &=\frac{1}{k!} \int_{\R}e^{-\alpha x}   x^k\nu_U(dx)
= \frac{1}{k!} \int_{\R}e^{-\alpha x}   x^k \kappa x^{-1}e^{- x}dx
= \frac{\kappa}{k!}\int_{\R}e^{-(\alpha+1)x}x^{k-1}dx\\
&=\frac{\kappa}{k}(\alpha+1)^{-k},\\
q_k^{(V_i)} &= \int_{\R} \frac{x^k}{k!}e^{-x}\nu_{V_i}(dx)
=\int_{\R} \frac{x^k}{k!}e^{-x}\kappa_ix^{-1}e^{-x/\beta_i}dx
=\frac{\kappa_i}{k!}\int_{\R}e^{-(1+1/\beta_i)x}x^{k-1}dx\\
&= \frac{\kappa_i}{k}(1+1/\beta_i)^{-k}.
\end{align*}
Recall the series expansion of the logarithm: $\sum_{k=1}^{\infty}\frac{x^k}{k} = -\log(1-x)$,
for $x\leq 1$ and $x\not = 1$. Hence we conclude that 
\begin{align*}
\mathcal{L}_{{\bf C}}({\boldsymbol \theta}) &= \frac{1}{v}\left\{\sum_{k=1}^{\infty} \left(\sum_{i=1}^n \alpha_ie^{-\theta_i}\right)^k q_k^{(U)}
+\sum_{i=1}^n\sum_{k=1}^{\infty} e^{-\theta_i k} q_k^{(V_i)} \right\}\\
&= \frac{1}{v}\left\{\sum_{k=1}^{\infty} \left(\sum_{i=1}^n \alpha_ie^{-\theta_i}\right)^k \frac{\kappa}{k}(\alpha+1)^{-k}
+\sum_{i=1}^n\sum_{k=1}^{\infty} e^{-\theta_i k} \frac{\kappa_i}{k}(1+1/\beta_i)^{-k}\right\}\\
&=\frac{1}{v}\left\{
-\kappa\log\left(1- \sum_{i=1}^n \frac{\alpha_i}{\alpha+1}e^{-\theta_i}\right)
-\sum_{i=1}^n\kappa_i \log\left(1-e^{-\theta_i }(1+1/\beta_i)^{-1} \right) 
\right\}.
\end{align*}

I.e.~we can either represent the distribution by  one discrete compound Poisson distribution. Alternatively, we can write it as convolution of $n+1$ independent compound Poisson laws, where one component has the 
 multivariate logarithmic distribution with parameters $(p_1,\dots,p_n)$ for  $p_i =\alpha_i/(\alpha+1)$ as the jump size distribution, see e.g.~\cite{GB1967} and Remark \ref{MLSD} below. The remaining components  have a one-dimensional logarithmic distribution in one component of the jump sizes and the other components are set to zero, more precisely, we can write
 \begin{align*}
 {\bf L}'_1 = \sum_{j=1}^{N_1^{||}}{\bf C}_j^{||} + \sum_{i=1}^n \sum_{j=1}^{N_1^{(i)\bot }} (0,\dots, 0, C_j^{(i)\bot }, 0, \dots, 0)^{\top},
 \end{align*}
 where the component $C_j^{(i)\bot }$ is in the $i$th row in the $n$-dimensional column vector.
 The Poisson random variable $N_1^{||}$ has intensity $\kappa \log(1+\alpha)$ and the Poisson random variables  $N_1^{(i)\bot }$ have rates $\kappa_i \log(1+\beta_i)$. Further, $C_j^{(i)\bot }\sim \text{Log}(\frac{\beta_i}{1+\beta_i})$.
\end{example}

\begin{remark}\label{MLSD}
Recall the following properties of the multivariate logarithmic series distribution, see e.g.~\cite{GB1967}. 
${\bf C}^{||} \sim \text{Log}(p_1,\dots,p_n)$, where $0<p_i<1, p:=\sum_{i=1}^np_i <1$ if for ${\bf c}\in \N_0^n \setminus \{{\bf 0} \}$,
\begin{align*}
\mathbb{P}({\bf C}^{||}={\bf c})=\frac{\Gamma(c_1+\cdots+c_n)}{c_1! \cdots c_n!}\frac{p_1^{c_1}\cdots p_n^{c_n}}{[-\log(1-p)]}.
\end{align*}
Each component $C^{||(i)}$ follows the modified univariate logarithmic distribution with parameters $\tilde p_i = p_i/(1-p+p_i)$ and $\delta_i= \log(1-p+p_i)/\log(1-p)$, i.e.
\begin{align*}
\mathbb{P}(C^{||(i)} =c_i) = \left\{ \begin{array}{ll} \delta_i, & \text{for } c_i=0\\
(1-\delta_i) \frac{1}{c_i}\frac{\tilde p_i^{c_i}}{[-\log(1-\tilde p_i)]}, & \text{for } c_i \in \N. \end{array} \right.
\end{align*}
\end{remark}

\section{Simulation and inference}\label{Section:SimInf}
We will now turn our attention to simulation and inference for trawl processes.  
We will start off by deriving a simulation algorithm which is based on the compound-Poisson-type representation of MIVTs. This will enable us to simulate sample paths from our new class of processes, which can be used for model-based parametric bootstrapping  in parametric inference. The inference procedure itself will be based on the (generalised) method of moments, since the cumulants of the multivariate trawl process are readily available.  
\subsection{Simulation algorithm}\label{simalgo}
First of all, we discuss how to simulate a univariate MIVT process. For each component 
 $i\in \{1,\dots,n\}$, we have the following representation. 
\begin{align*}
Y_t^{(i)} &= L^{(i)}(A^{(i)}_t) =X_{0,t}^{(i)}+ X_t^{(i)},
\end{align*}
where
$X_{0,t}^{(i)}=L^{(i)}(\{(x, s): s \leq 0, 0 \leq  x \leq  d^{(i)}(s-t)\})$ and $X_t^{(i)}= L^{(i)}(\{(x, s): 0<s \leq t, 0 \leq  x \leq  d^{(i)}(s-t)\})$, for a  
trawl function $d^{(i)}$. 

We would like to argue that the term $X_{0,t}^{(i)}$ is asymptotically negligible in the sense that it converges to zero  as $t\to\infty$, which will allow us to concentrate on the term $X_t^{(i)}$ in the following. Indeed, this conjecture holds as the following proposition shows.
\begin{proposition}\label{Prop_ConvofIniValue}
For a 
trawl function $d^{(i)}$, we have that $X_{0,t}^{(i)}=L^{(i)}(\{(x, s): s \leq 0, 0 \leq  x \leq  d^{(i)}(s-t)\})\to 0$ in probability, as $t\to \infty$.
\end{proposition}

Hence, we will focus on simulating $X_t^{(i)}$ and will work with a burn-in period in the simulation such that the effect of $X_{0,t}^{(i)}$ is negligible.

A realisation of ${\bf L}$ consists of a countable set $R$ of points $({\bf y},x,s)$ in
$\N_0^n\setminus\{ {\bf 0}\} \times [0,1]\times \R$.
When we project the point pattern to the time axis, we obtain the arrival times of a Poisson process $N_t$ with intensity $v= \nu(\R^n)$. The corresponding arrival times are denoted by $t_1, \dots, t_{N_t}$ and we associate uniform heights $U_1, \dots, U_{N_t}$ with them, see \cite{BNLSV2014} for a detailed discussion in the univariate case. 
So as soon as we have specified the jump size distribution of the ${\bf C}$, we can use the representation  
\begin{align*}
X_t^{(i)}=  \sum_{j=1}^{N_t}C_j^{(i)}\mathbb{I}_{\{U_j \leq d^{(i)}(t_j-t)\}},
\end{align*}
 to simulate each component.

\begin{algorithm} \label{GenSimAlgo}
In this algorithm we suppress the dependence on the superscript ${(i)}$ and describe how to simulate from the one-dimensional  of components of the form
\begin{align}\label{CPP}
X_t:=\sum_{j=1}^{N_t}C_j\mathbb{I}_{\{U_j \leq d(t_j-t)\}}
\end{align}
We want to simulate $X$ on a $\Delta$-grid of $[0,t]$, where $\Delta >0$, i.e.~we want to find $(X_0,X_{1\Delta}, \dots, X_{\lfloor t/\Delta \rfloor\Delta})$.
\begin{enumerate}
\item Generate a realisation $n_t$ of the the Poisson random variable $N_t$ with mean $v t$ for $v>0$.
\item Generate the pairs $(t_j, U_j)_{j\in \{1,\dots,n_t\}}$ where the series $(t_1,\dots, t_{n_t})$ consists of realisations of ordered i.i.d.~uniform random variables on $[0,t]$. The $(U_1,\dots,U_{n_t})$ are i.i.d.~and uniformly distributed on $[0,1]$ and independent of the arrival  times $(t_1,\dots, t_{n_t})$.
\item Simulate the i.i.d.~jump sizes $C_1,\dots, C_{n_t}$.
\item Construct  the trawl process on a $\Delta$-grid, where $\Delta >0$, by setting
$X_0=0$ and 
\begin{align}\label{simstep}
X_{k\Delta}:=\sum_{j=1}^{\text{card}\{t_l:t_l\leq k\Delta \}}C_j\mathbb{I}_{\{U_j \leq d(t_j-k\Delta)\}}, \quad k=1, \dots, \lfloor t/\Delta \rfloor.
\end{align}
\end{enumerate}
\end{algorithm}
\begin{remark}
Note that the condition in the indicator function in \eqref{simstep} can be expressed in a vectorised form, which allows a fast implementation of the simulation algorithm, see Section \ref{SimTrawlAppendix} for details. 
\end{remark}

In order to generate samples from the multivariate process, it is easiest to split the compound Poisson seed into dependent and independent components and simulate the components separately as we shall describe in more detail in the following example.
\begin{example}
Suppose we want to simulate from the multivariate trawl process with negative binomial marginal law as described in Example \ref{ex_indanddep}.
Then we split each component into a dependent and an independent component as follows:
\begin{align*}
X_t^{(i)}= \sum_{j=1}^{N_t}C_j^{(i)}\mathbb{I}_{\{U_j \leq d^{(i)}(t_j-t)\}} =
\sum_{j=1}^{N_t^{||}}C_j^{||(i)}\mathbb{I}_{\{U_j \leq d^{(i)}(t_j-t)\}} 
+\sum_{j=1}^{N_t^{(i)\bot }}C_j^{\bot (i)}\mathbb{I}_{\{U_j \leq d^{(i)}(t_j-t)\}},
\end{align*}
where
${\bf C}^{||} =(C_j^{||(1)}, \dots,C_j^{||(n)})^{\top} \sim \text{Log}(p_1,\dots,p_n)$, where $p_i=\alpha_i/(1+\alpha)$, $C_j^{\bot (i)} \sim \text{Log}(p_i)$. Note that ${\bf C}^{||}$ and  $C_j^{\bot (1)}, \dots, C_j^{\bot (n)}$ are independent for all $j$ and the intensities of the independent Poisson processes  $N^{||}, N^{(1)\bot },\dots,N^{(n)\bot }$ are given by $\kappa \log(1+\alpha)$ and $\kappa_1 \log(1+\alpha_1), \dots, \kappa_n \log(1+\alpha_n)$, respectively.
Then we can use the algorithm above to simulate each component separately.
\end{example}

\begin{remark}
Since the above scheme ignores the initial value $X_{0,t}$, it is advisable to work with a burn-in period in a practical implementation. In the situation when  the support of the trawl function $d$ is bounded, then an exact simulation of the trawl process is possible since its initial value can be generated precisely.
\end{remark}

\subsection{Inference}\label{InfSection}
We propose to estimate the model parameters using a two stage equation-by-equation procedure, where the marginal parameters for each component are estimated first, and the parameter determining the dependence are estimated in a second step. Recent research on inference in multivariate models, see e.g.~\cite{Joe2005} and, more recently, \cite{FrancqZakoian2016}, has highlighted that such a procedure is very powerful in a   high-dimensional set-up.

Motivated by the results in \cite{BNLSV2014}, we propose to work with the (generalised) method of moments to infer the model parameters since the cumulants are readily available and the procedure works well in our simulation study. Full maximumlikelihood estimation is numerically rather intractable, whereas composite likelihood methods based on pairwise observations also seem to work well in the univariate case, as  ongoing work not reported here, reveals.  

In the following, we shall assume that we have decided on a  parametric model for the multivariate trawl process with trawl   functions $d^{(i)}$. 

{\bf Step 1:} 
We can  use the time series for each component to estimate the marginal parameters. We will estimate the parameter of $d^{(i)}$ in Step a) and the ones of $L'^{(i)}$ in Step b).

{\bf a)}
Recall that for each component we 
have the following representation for the autocorrelation function:
\begin{align*}
r_{ii}(h)=\Cor\left(L^{(i)}(A_t^{(i)}), L^{(i)}(A_{t+h}^{(i)})\right)
= \frac{\leb(A^{(i)}\cap A^{(i)}_h)}{\leb(A^{(i)})}=\frac{R_{ii}(h)}{\leb(A^{(i)})}, \quad \text{for }i = 1,\dots,n.
\end{align*}
I.e.~the autocorrelation function only depends on the parameters of the trawl function $d^{(i)}$. These parameters  can hence be estimated 
by using the method of moments or generalised method of moments (depending on the model specification) by matching the empirical and the theoretical autocorrelation function. This will, in particular, provide us with an estimate of $\leb(A^{(i)})$.

{\bf b)}
In a second step, we can then estimate the parameters determining the marginal distribution of $L^{'(i)}$, again using a method of moments,  by using a sufficient number of cumulants of the 
observed trawl process. Note that 
the cumulant function for an individual component has the form 
\begin{align}\label{CumTrawl}
C(\xi \ddagger Y_t^{(i)}) = \leb(A^{(i)}) C(\xi \ddagger L^{'(i)}),  \quad \text{ for } i =1, \dots, n,.
\end{align}
I.e.~the cumulants of the trawl process can be easily derived. 
We denote by $\kappa_k$ the $k$th cumulant for $k \in \mathbb{N}$. Then we have that 
\begin{align*}
\kappa_k(Y_t^{(i)}) = \leb(A^{(i)})\  \kappa_k(L^{'(i)}), \quad \text{ for } i =1, \dots, n.
\end{align*}
I.e.~as long as the parameters are identified through the cumulants, we can estimated them after having estimated the trawl parameters by setting 
\begin{align*}
\widehat{\kappa_k}(L^{'(i)}) = \frac{\kappa_k^e(Y_t^{(i)})}{\widehat{\leb(A^{(i)})}}  , \quad \text{ for } i =1, \dots, n,
\end{align*}
where $\kappa_k^e$ stands for the corresponding empirical $k$th cumulant. We then just need to solve  the equations for the corresponding parameters.  If a direct matching does not work, then one can use the generalised method of moments. 

{\bf Step 2:} After the marginal parameters have been identified, we turn to estimating the parameters describing the dependence. 
We note that as soon as the trawl parameters have been estimated, the corresponding autocorrelators can be computed. I.e.~we then obtain estimates $\widehat{R}_{ij}(h)=\leb(A^{(i)}\cap A^{(j)}_h)$.  
Then we get that
\begin{align*}
\kappa_{i,j} = \frac{\rho_{ij}^e(h)}{\widehat{R}_{ij}(h)}, \quad \text{ for } i\not = j, i,j \in \{1, \dots, n\},
\end{align*}
where $\rho_{ij}^e(h)$ denote the empirical autocovariance function between the $i$th and $j$th component evaluated at lag $h$. In fact, it will be sufficient to set $h=0$ when we estimate the parameters $\kappa_{i,j}$.
Note here that while we can clearly estimate the pairwise covariance parameter  $\kappa_{i,j}$ using this method, depending on the parametric model chosen, there might be more than one parameter describing the dependence structure. As such, not all parameters might be identified through this procedure in which case additional moment conditions need to be considered. However, since this scenario did not arise in the model specifications we studied in relation to our empirical work, we shall refer this aspect to future research. 

Let us briefly comment on the validity of this estimation method: According to  \cite{FuchsStelzer2013} multivariate mixed moving-average
processes are mixing as long as they exist. This result implies that our stationary multivariate  trawl processes are mixing and   hence also weakly mixing and ergodic.
Hence we can deduce that   moment-based estimation methods are consistent, see e.g. \cite{Matyas1999}.

In order to construct confidence bounds for the various parameters, we proceed by implementing a parametric bootstrap procedure, where we plug in the estimated parameters into the model specification, simulate from the model as described in the previous section, and then  report the corresponding 95\% confidence bounds. 
\subsubsection{Simulation study}\label{SimStudy}
In order to check how well the inference procedure works in finite samples, we conduct a Monte Carlo study, where we choose the model setting which describes our empirical data well, see Section \ref{Section:Empirics}. To this end, we simulate  samples consisting of 3960 observations each from the bivariate version of the negative binomial model with common factor as described in Example \ref{ex_dep}. The distribution of the corresponding \Levy seed is determined by three parameters: $\alpha_1, \alpha_2$ and $\kappa$. In addition, we choose an exponential trawl function for both components, which are parametrised by  $\lambda_1$ and $\lambda_2$, respectively.  The parameters are set to their empirical counterparts, see Table 
\ref{table_par} below. Some of the technical details regarding the simulation study can be found in Section \ref{sim} in the Appendix.

We draw 5000 samples from the model using the simulation algorithm described above and estimate the parameters for each sample using the method of moments. In Figure \ref{boxplots}, we present the boxplots for the estimates for each of the five parameters. The true values are highlighted by a vertical red line. We observe that all five estimates center around the true values. 
\begin{figure}[htbp]
  \includegraphics[width=0.55\textwidth, height=0.7\textheight,angle=270]{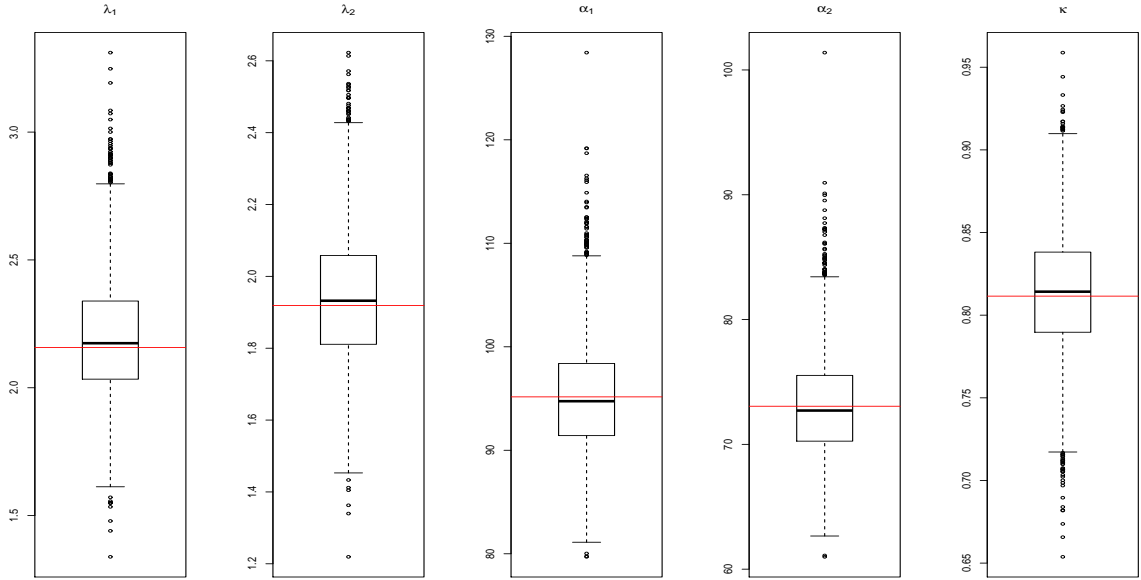}
  \caption{Boxplots of the five parameter estimates from a bivariate trawl model with exponential trawl function and negative binomial \Levy seed, see Example \ref{ex_dep}. The results are based on 5000 Monte Carlo runs, where each sample contains 3960 observations. The true values are indicated by a red vertical line. \label{boxplots}}
\end{figure}

\section{Empirical illustration}\label{Section:Empirics}
In this section, we apply our new modelling framework to high frequency financial data. More precisely, we study limit order book data from the database LOBSTER\footnote{LOBSTER: Limit Order Book System - The Efficient Reconstructor at Humboldt Universit\"{a}t zu Berlin, Germany. http://LOBSTER.wiwi.hu-berlin.de}.

We have downloaded the limit order book data for 
Bank of America (ticker: BAC) for one day (21st April 2016).
We are interested in investigating the joint behaviour between the time series of the number of newly submitted limit orders versus the number of fully deleted limit orders.  Note that the trading day starts at 9:30am and ends at 16:00. For our analysis, we discard the first and last 30 minutes of the data which typically have a peculiar (non-stationary) structure due to the effects caused by the beginning and the end of trading. As such we analyse data for a time period of 5.5 hours. We split this time period  into intervals of length five seconds, resulting in 3960 intervals. In each interval we count the number of newly submitted limit orders and the ones which have been fully deleted.

\begin{table}[htbp]
\center
\begin{tabular}{ccccccc}
\hline
 & Min & 1st Quartile &  Median &  Mean & 3rd Quartile & Max \\
\hline
  No.~of new submissions &   0&    7 &   13 &   34.06 &   28 & 646 \\
 No.~of full deletions &     0&   5.75 &   12 &   29.13 &   27.25 &  571  \\
\hline
\end{tabular}
\caption{Summary statistics of the  
BAC data from 21st April 2016 
based on   intervals of length 
five seconds.  Also, we find that the correlation between the two time series is equal to  0.984.
\label{table_subsum}}
\end{table}

 The summary statistics of these two count series are provided in 
Table \ref{table_subsum}. Moreover, Figure \ref{TSPlot} depicts the corresponding time series plot, which also includes a picture of the difference of the two time series (in the middle), and Figure \ref{Empirical_Histogram}  presents histograms of the joint and the marginal distribution of the data.

 \begin{figure}[htbp] 
  \includegraphics[width=0.55\textwidth, height=0.7\textheight,angle=270]{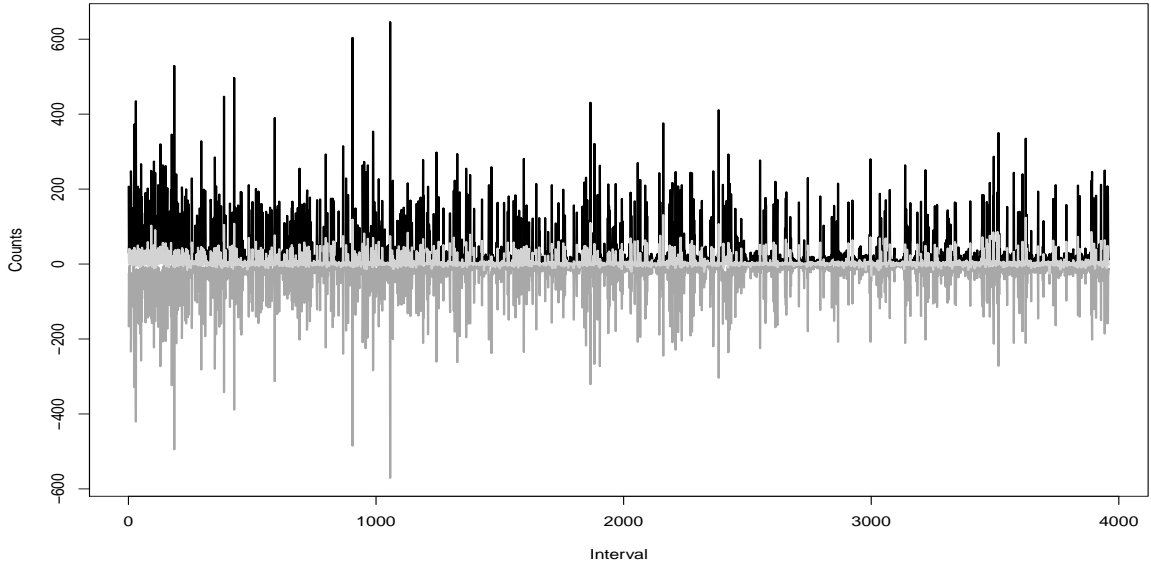}
  \caption{Time series plots of the BAC data from 21st April 2016 
based on   intervals of length 
five seconds.  Black (top): number of submitted orders; light grey (middle): number of  submitted - fully deleted orders; dark grey (bottom):  - number of  fully deleted orders. \label{TSPlot}}
\end{figure}

We observe that there is a very strong correlation and co-movement between the two time series, which confirms the well-known fact that, for highly traded stocks such as  BAC,  the majority of newly submitted limit orders gets deleted rather than executed.

\begin{figure}[htbp] \centering
  \includegraphics[scale=0.4,angle=270]{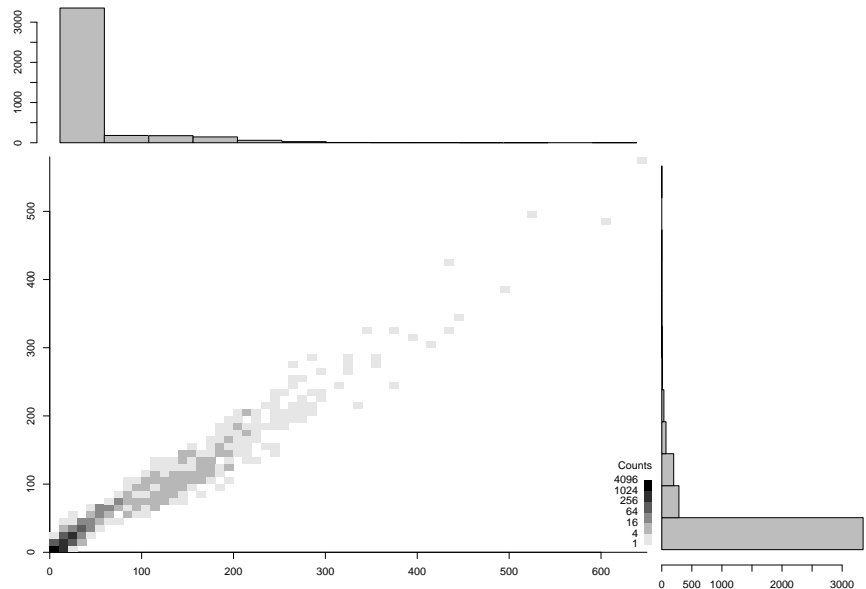}
  \caption{Histograms depicting the joint distribution and the marginal distributions of the new submissions (top) and the full cancellations (right). \label{Empirical_Histogram}}
\end{figure}

 \begin{figure}[htbp]
  \subfigure[ACF of new submissions]{
\includegraphics[scale=0.25,angle=270]{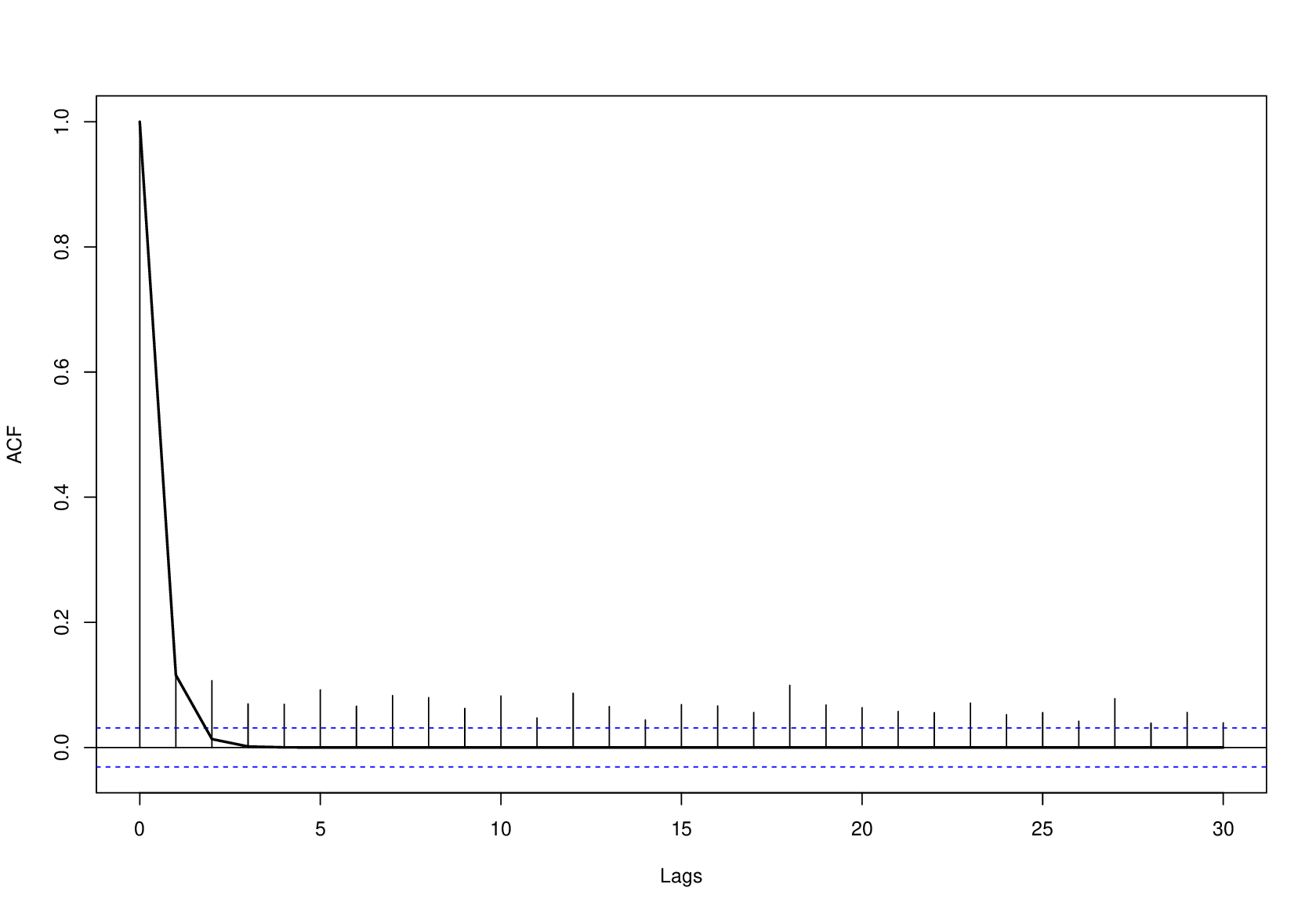}
 }
 \subfigure[ACF of full cancellations.]{
\includegraphics[scale=0.25,angle=270]{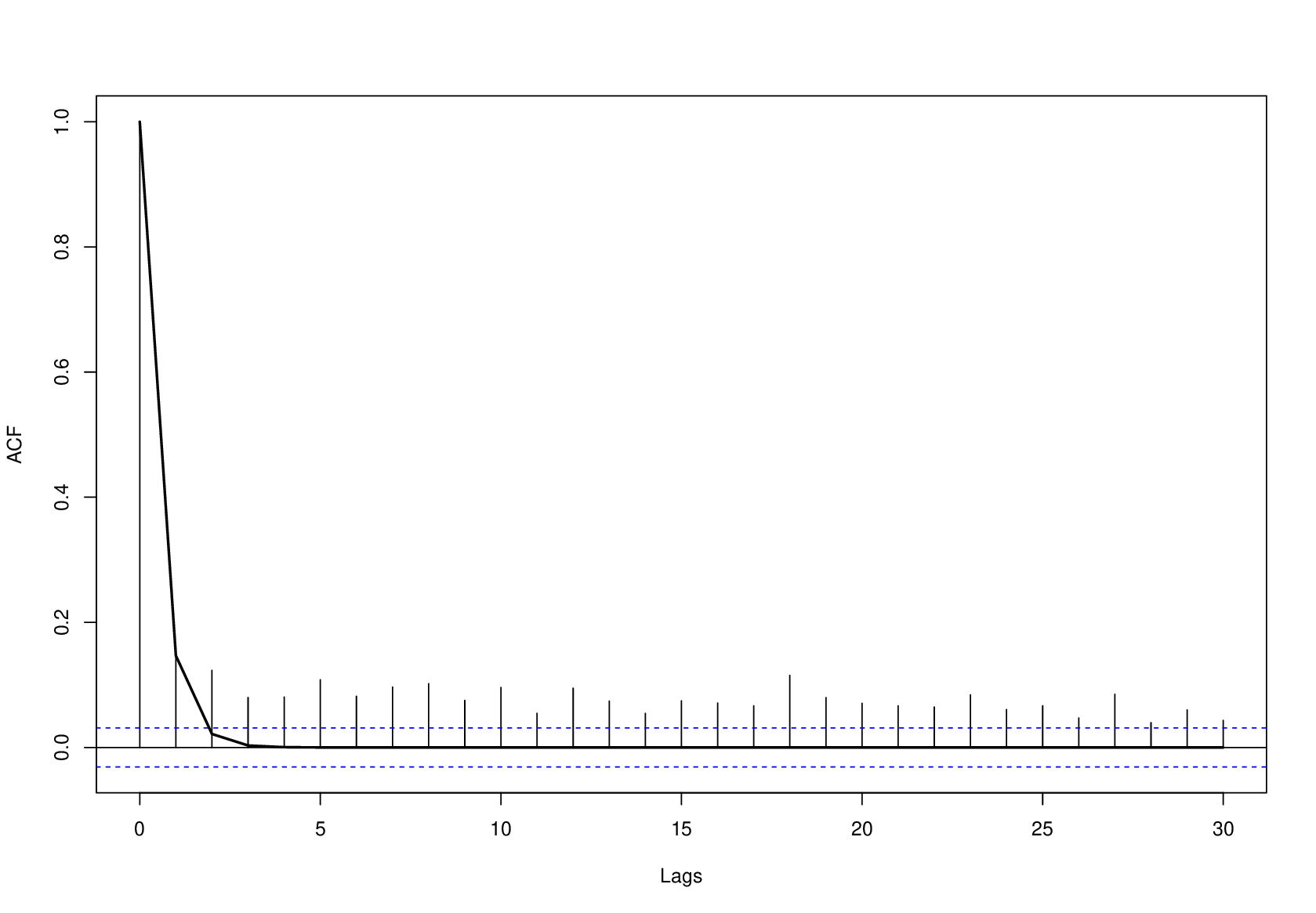}
 }\caption{Empirical autocorrelation function (ACF) of the number of newly submitted and fully deleted limit orders, respectively. The solid black line  shows the estimated exponential trawl function in both cases. \label{Fig-ACF}}
 \end{figure}

Since the empirical autocorrelation function decays rather quickly for both time series,  see  Figure \ref{Fig-ACF},  we fit an exponential trawl function in both cases and get a good fit. Based on the estimated trawl parameters, we compute $\leb(A^{(1)})$, $\leb(A^{(2)})$, and $\leb(A^{(1)}\cap A^{(2)})$. 
Next, we estimate the parameters $\alpha_1, \alpha_2$ from the marginal law and finally infer $\kappa$ from the empirical cross-covariance. All parameter estimates are summarised in Table  \ref{table_par}. In addition, we provide the corresponding 95\% confidence intervals, which are based on a parametric bootstrap, where we simulated 5000 samples from the estimated model using the estimated parameters as the plug-in values.

  \begin{table}[h]
\center
\begin{tabular}{c|cc|ccccc}
\hline
  & $\lambda_1$ & $\lambda_2$ 
&  $\alpha_1$ & $\alpha_2$ 
 & $\kappa$  \\ \hline
Estimates & 2.157
& 1.919
&95.161& 73.055
& 0.812\\ 
CB & {\footnotesize(1.771, 2.673)} & {\footnotesize(1.597, 2.322)} 
& {\footnotesize(85.321, 106.147)} & {\footnotesize(65.797, 81.222) }
& {\footnotesize(0.741, 0.885)} \\
\hline
\end{tabular}
\caption{Estimated parameters and estimates of the 95\% confidence bounds (CB)  from the moment-based estimates. The CB estimates have been computed using a model-based bootstrap, where 5,000 bootstrapped samples were drawn. \label{table_par}}
\end{table}

\begin{figure}[htbp]
  \includegraphics[scale=0.5,angle=270]{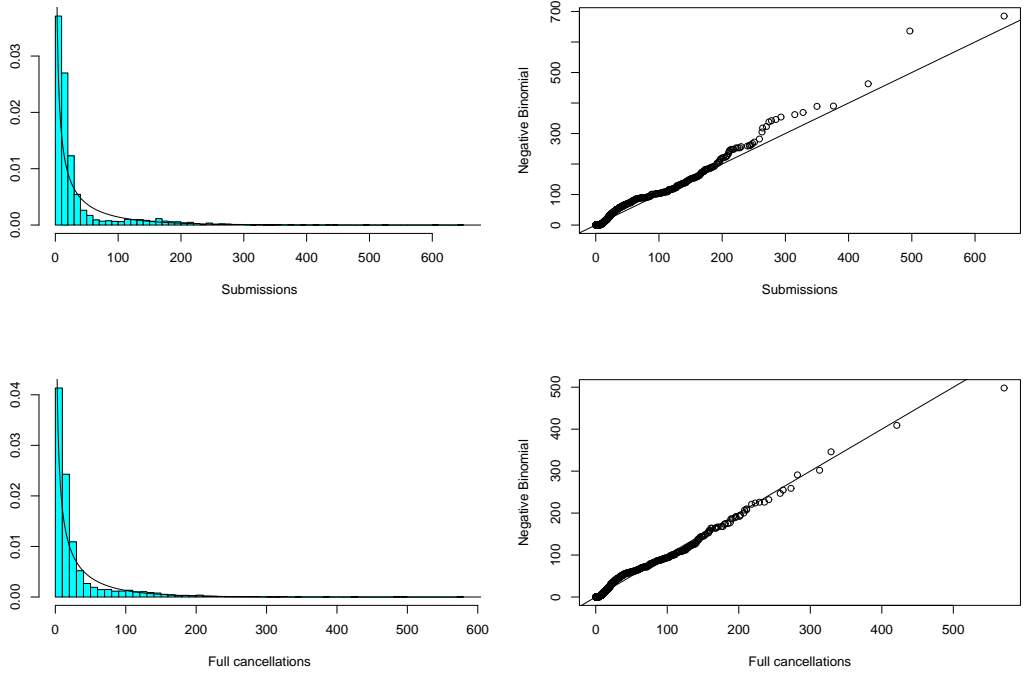}
  \caption{Empirical and fitted densities and quantile-quantile plots of the negative binomial marginal law for the new submissions (top) and the full cancellations (bottom). \label{GoodnessOfFit}}
\end{figure}

\begin{figure}[htbp] \centering
  \includegraphics[scale=0.4,angle=270]{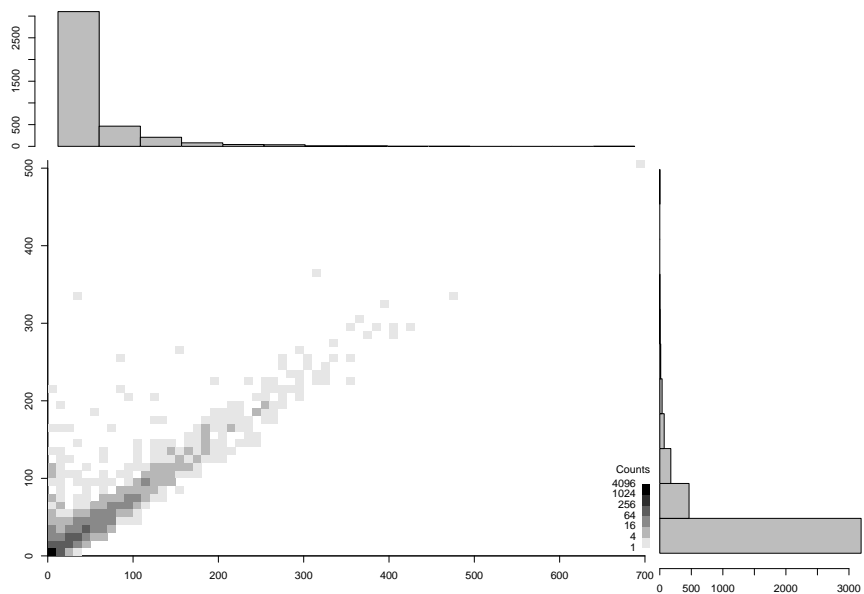}
  \caption{Histograms depicting the joint distribution and the marginal distributions of one path of the simulated bivariate time series in our bootstrap procedure.   \label{Sim_Histogram}}
\end{figure}

 \begin{figure}[htbp] 
  \includegraphics[width=0.55\textwidth, height=0.7\textheight,angle=270]{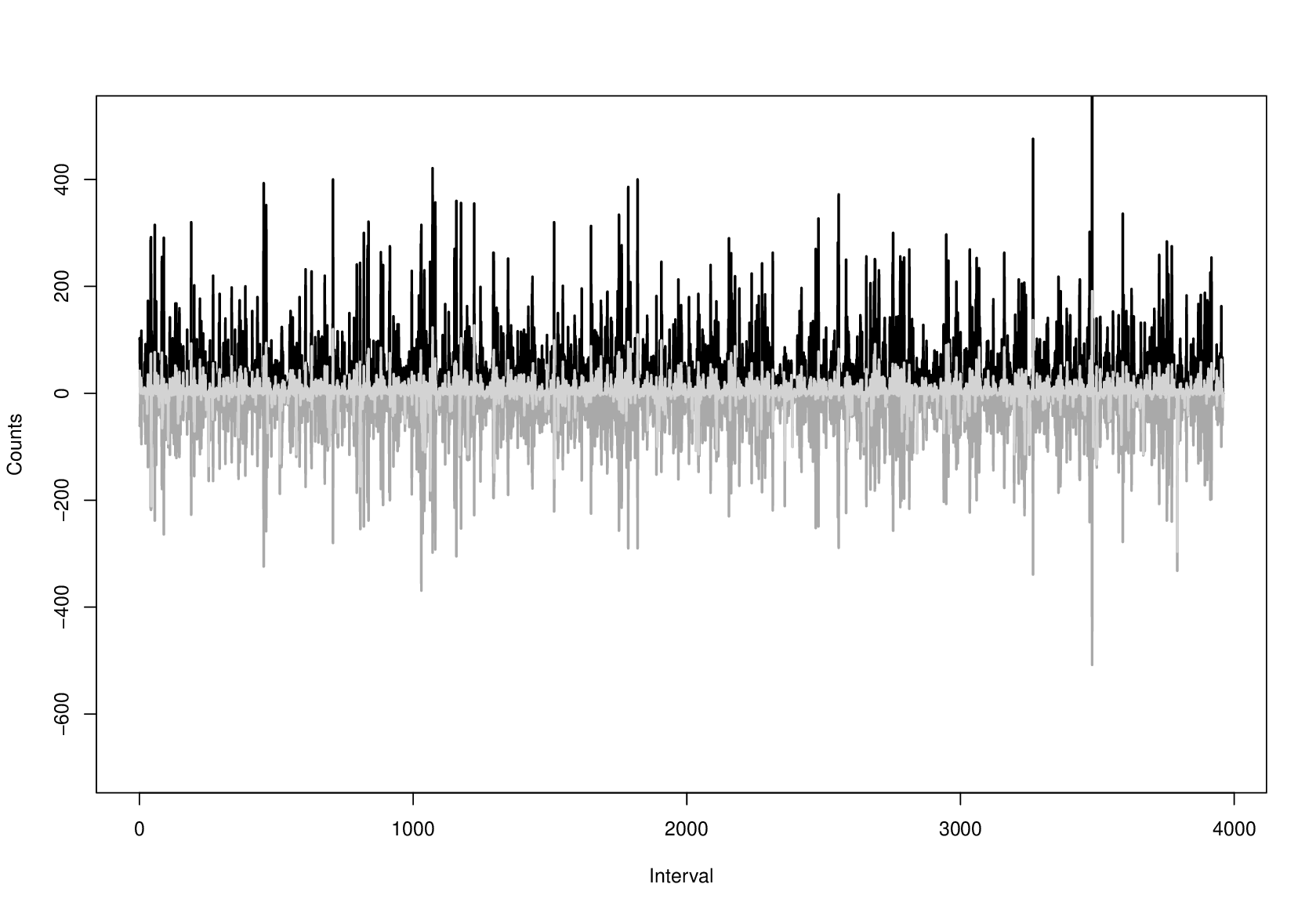}
  \caption{Time series plots of one simulated sample path. Black (top): first component; light grey (middle): number of  first - second component; dark grey (bottom):  - second component. This is the same path as the one used to generate Figure \ref{Sim_Histogram}.\label{TSPlotSim}}
\end{figure}

 In addition to checking the goodness-of-fit of the trawl function, see Figure \ref{Fig-ACF}, we also need to assess whether the parametric model for the bivariate \Levy seed is appropriate.  To this end, we first check the marginal fit, which corresponds to a univariate negative binomial law for each component.
Figure \ref{GoodnessOfFit} shows the empirical and the estimated probability densities and the corresponding quantile-quantile plots. While the fit seems to be acceptable overall, we  note that the fit appears to be better for the time series of the cancelled orders, where the quantile-quantile-plot is closer to a straight line, than in the case of the newly submitted orders, where we observe a mildly wiggly line.
Finally, we investigate the goodness-of-fit of the joint law. For this, we draw the bivariate law from one of our bootstrap samples and the corresponding univariate laws, see Figure \ref{Sim_Histogram}. We observe that the histogram of the simulated  joint law resembles the one from the empirical data well, cf.~Figure \ref{Empirical_Histogram}. Also a visual inspection of the simulated sample paths, see Figure \ref{TSPlotSim} for one example, shows that the empirical data and the simulated data have indeed very similar features, which supports our hypothesis that a bivariate trawl process can describe the number of order submissions and cancellations in a limit order book well.

\section{Conclusion}\label{Section:conclusion}
We propose a new modelling framework for multivariate time series of counts, which is based on so-called multivariate integer-valued trawl (MVIT) processes. 
Such processes are
highly analytically tractable and enjoy useful properties, such as 
 stationarity, infinitely divisibility, ergodicity and a  mixing  property. A variety of serial dependence patterns, including short and long memory, as well as all discrete infinitely divisible marginal distributions can be achieved within this novel framework.  
In this article, we focused in particular on various specifications of a multivariate infinitely divisible negative binomial distribution, since its univariate counterpart has been widely used in  empirical work.
Moreover, since the MVIT process is defined in continuous time, it can be applied to non-equidistant and asynchronous data, which increases its broad applicability.
Further contributions of this article include a simulation algorithm for MVIT processes and a suitable inference procedure which is based on the two-stage equation-by-equation approach, where the parameters describing the univariate marginal distributions  are estimated in the first step, followed by the estimation of the dependence parameters in the second step. A simulation study confirms the effectiveness of this inference method in finite samples. The estimation itself is based on the generalised method of moments and suitable confidence bounds are obtained through a parametric bootstrap procedure. 
In an empirical illustration, a bivariate version of an MVIT process has been used to successfully describe the relationship between the number of order submissions and cancellations in a limit order book.

\begin{appendix} 
\section{Proofs}\label{proofs}
\begin{proof}[Proof of Proposition \ref{CharFctTrawl}]
Using the properties of the \Levy basis, we immediately obtain that 
\begin{align*}
&\mathbb{E}(\exp(i{\boldsymbol \theta}^{\top}{\bf Y}_t)) =
\exp\left(\int_{\R^n\times [0,1]\times \R} \left\{\exp\left(i\sum_{j=1}^n\theta_j\mathbf{I}_{A^{(j)}}(x,s-t)y_j\right)-1\right\} \nu(d{\bf y})  dx ds\right).
\end{align*}
The expression for the characteristic function can be further simplified by using
a partition $S =\{ S_1, \dots, S_{2^n-1}\}$ of  $A^{\cup,n }:=\cup_{i=1}^n A^{(i)}$, see  \cite{NVG2015}.
More precisely, we have that 
 \begin{align}\label{SP}
 A^{\cup,n} = \bigcup_{k=1}^{n}\bigcup_{\substack{1\leq i_1, \dots, i_k \leq n: \\ i_{\nu}\not = i_{\mu}, \text{ for } \nu \not = \mu}} \left(\bigcap_{l=1}^kA^{(i_l)}\setminus \bigcup_{\substack{1\leq j\leq n,\\ j\not \in\{i_1, \dots, i_k\}}} A^{(j)}\right).
 \end{align}
Note that 
\begin{align}\label{thetaY}
{\boldsymbol \theta}^{\top}{\bf Y}_t = \sum_{j=1}^n \theta_j L^{(j)}(A^{(j)}_t) = \sum_{j=1}^n \theta_j \sum_{k: S_k\subset A^{(j)}} L^{(j)}(S_k) = \sum_{k=1}^{2^n-1}   \sum_{\substack{1\leq j\leq n:\\  A^{(j)}\supset S_k}} \theta_j L^{(j)}(S_k).
\end{align}
Finally, combining \eqref{thetaY} with the  representation \eqref{SP} and using the fact that a \Levy basis is  independently scattered, we obtain the result.
\end{proof}

\begin{proof}[Proof of Proposition \ref{JointLaw}] The joint law is given by 
\begin{align*}
&P(X_1=x_1, \dots, X_n=x_n) \\
&=
\int_{(0,\infty)^{n+1}}P(X_1=x_1, \dots, X_n=x_n|U=u, V_1=v_1, \dots, V_n=v_n)\\
& \qquad \qquad \qquad \qquad \qquad \qquad \qquad \qquad \cdot f_U(u)f_{V_1}(v_1)\cdots f_{V_n}(v_n)du dv_1 \cdots dv_n\\
&= \int_{(0,\infty)^{n+1}} \prod_{i=1}^n e^{-\alpha_i u + v_i}\frac{(\alpha_i u + v_i)^{x_i}}{x_i!} f_U(u)f_{V_i}(v_i)du dv_i\\
&=
\int_{(0,\infty)^{n+1}}  f_U(u)\prod_{i=1}^n e^{-\alpha_i u + v_i}\frac{1}{x_i!} \sum_{j_i=0}^{x_i}{x_i \choose j_i}\alpha_i^{j_i}u^{j_i}v_i^{x_i-j_i} f_{V_i}(v_i)du dv_i\\
&= \frac{1}{x_1!\cdots x_n!} \sum_{j_1=0}^{x_1}\cdots \sum_{j_n=0}^{x_n} {x_1 \choose j_1}\cdots{x_n \choose j_n}\alpha_1^{j_1}\cdots \alpha_j^{j_n} \E(U^{j_1+\cdots+j_n}e^{-(\alpha_1 + \cdots + \alpha_n)U}) \\
& \qquad \qquad \qquad \qquad\qquad \qquad \qquad \qquad\cdot \prod_{k=1}^n \E(V_k^{ x_k-j_k}e^{-V_k}).
\end{align*}
\end{proof}

\begin{proof}[Proof of Proposition \ref{CPRep}]
Let $M_U$ and $M_{V_i}$ denote the moment generating functions of $U$ and $V_i$, respectively.
According to \citet[equation (5.1)]{BNBlaesildSeshadri1992}, the probability generating function of $(X_1, \dots, X_n)$ is given by
\begin{align*}
G(t_1, \dots, t_n) &= E(t_1^{X_1} \cdots t_n^{X_n})= M_U\left(\sum_{i=1}^n \alpha_i(t_i-1)\right)\prod_{i=1}^n M_{V_i}(t_i-1).
\end{align*}
Hence, the corresponding Laplace transform for positive $\boldsymbol \theta$ is given by 
\begin{align}\label{LaplaceCPP2}
\mathcal{L}(\theta_1,\dots,\theta_n)&= G(e^{-\theta_1}, \dots,e^{-\theta_n})
=
M_U\left(\sum_{i=1}^n \alpha_i(e^{-\theta_i}-1)\right)\prod_{i=1}^n M_{V_i}(e^{-\theta_i}-1).
\end{align}

The aim is to find $v$ and $\mathcal{L}_{{\bf C}}({\boldsymbol \theta})$ by equating \eqref{LaplaceCPP1} and \eqref{LaplaceCPP2}.
Using the relation between the Laplace and the moment generating function, we deduce that 

\begin{align*}
\mathcal{L}(\theta_1,\dots,\theta_n)&= 
M_U\left(\sum_{i=1}^n \alpha_i(e^{-\theta_i}-1)\right)\prod_{i=1}^n M_{V_i}(e^{-\theta_i}-1)\\
&= \mathcal{L}_U\left(\sum_{i=1}^n \alpha_i(1-e^{-\theta_i})\right)\prod_{i=1}^n \mathcal{L}_{V_i}(1-e^{-\theta_i})\\
&= \exp\left(\log\mathcal{L}_U\left(\sum_{i=1}^n \alpha_i(1-e^{-\theta_i})\right)+\sum_{i=1}^n \log\mathcal{L}_{V_i}(1-e^{-\theta_i})\right).
\end{align*}
We use the notation $\overline K = \log \mathcal{L}$ for the so-called \emph{kumulant function}.
Since $U$ is a subordinator without drift, we have that 
\begin{align*}
&\overline K_U\left(\sum_{i=1}^n \alpha_i(1-e^{-\theta_i})\right)
= \int_{\R}\left(e^{-\sum_{i=1}^n \alpha_i(1-e^{-\theta_i})x}-1 \right)\nu_U(dx)\\
&=\int_{\R}\left(e^{-\sum_{i=1}^n \alpha_i x} -e^{-\sum_{i=1}^n \alpha_i x} + e^{\sum_{i=1}^n \alpha_i(1-e^{-\theta_i})x}-1 \right)\nu_U(dx)
\\
&=\int_{\R}\left(e^{-\sum_{i=1}^n \alpha_i x} -1 \right)\nu_U(dx)
+\int_{\R}e^{-\sum_{i=1}^n \alpha_i x} \left(e^{\sum_{i=1}^n \alpha_ie^{-\theta_i}x}-1 \right)\nu_U(dx).
\end{align*}
Note that
\begin{align*}
e^{\sum_{i=1}^n \alpha_ie^{-\theta_i}x}-1 &= \sum_{k=1}^{\infty}\frac{1}{k}\left(\sum_{i=1}^n \alpha_ie^{-\theta_i}x\right)^k
=\sum_{k=1}^{\infty}\frac{1}{k!}\left(\sum_{i=1}^n \alpha_ie^{-\theta_i}\right)^k x^k.
\end{align*}
We set $\alpha := \sum_{i=1}^n \alpha_i$. Then
\begin{align*}
\int_{\R}e^{-\sum_{i=1}^n \alpha_i x} \left(e^{\sum_{i=1}^n \alpha_ie^{-\theta_i}x}-1 \right)\nu_U(dx) 
&=
\int_{\R}e^{-\alpha x}  \sum_{k=1}^{\infty}\frac{1}{k!}\left(\sum_{i=1}^n \alpha_ie^{-\theta_i}\right)^k x^k\nu_U(dx)\\
&=  \sum_{k=1}^{\infty} \left(\sum_{i=1}^n \alpha_ie^{-\theta_i}\right)^k \underbrace{ \frac{1}{k!} \int_{\R}e^{-\alpha x}   x^k\nu_U(dx)}_{:=q_k^{(U)}}.
\end{align*}
I.e.
\begin{align*}
\overline K_U\left(\sum_{i=1}^n \alpha_i(1-e^{-\theta_i})\right)
&=\int_{\R}\left(e^{-\alpha x} -1 \right)\nu_U(dx)
+  \sum_{k=1}^{\infty} \left(\sum_{i=1}^n \alpha_ie^{-\theta_i}\right)^k q_k^{(U)}\\
&= \overline K_U(\alpha) +  \sum_{k=1}^{\infty} \left(\sum_{i=1}^n \alpha_ie^{-\theta_i}\right)^k q_k^{(U)}.
\end{align*}
Similarly, 
\begin{align*}
\sum_{i=1}^n K_{V_i}(1-e^{-\theta_i})&=\sum_{i=1}^n \left(\overline K_{V_i}(1)
+  \sum_{k=1}^{\infty} e^{-\theta_i k} q_k^{(V_i)}\right), \text{ where  } q_k^{(V_i)} = \int_{\R} \frac{x^k}{k!}e^{-x}\nu_{V_i}(dx).
\end{align*}
So, overall we have
\begin{align*}
\overline K_{\bf X}({\boldsymbol \theta}) &=\log \mathcal{L}(\theta_1, \dots, \theta_n) \\
&= \left( \overline K_U(\alpha) + \sum_{i=1}^n \overline K_{V_i}(1)\right) +  \sum_{k=1}^{\infty} \left(\sum_{i=1}^n \alpha_ie^{-\theta_i}\right)^k q_k^{(U)}
+\sum_{i=1}^n\sum_{k=1}^{\infty} e^{-\theta_i k} q_k^{(V_i)}\\
&= -v + v \mathcal{L}_{{\bf C}}({\boldsymbol \theta}), 
\end{align*}
if and only if
\begin{align*}
v &= -\left( \overline K_U(\alpha) + \sum_{i=1}^n \overline K_{V_i}(1)\right) ,\\
\mathcal{L}_{{\bf C}}({\boldsymbol \theta}) &= \frac{1}{v}\left\{\sum_{k=1}^{\infty} \left(\sum_{i=1}^n \alpha_ie^{-\theta_i}\right)^k q_k^{(U)}
+\sum_{i=1}^n\sum_{k=1}^{\infty} e^{-\theta_i k} q_k^{(V_i)} \right\}.
\end{align*}
\end{proof}
\begin{proof}[Proof of Proposition \ref{Prop_ConvofIniValue}]
The requirement that $\leb(A^{(i)})<\infty$ implies  that $\leb(\{(x, s): s \leq 0, 0 \leq  x \leq  d(s-t)\})\to 0$ as $t\to \infty$.
Since a  \Levy basis  is countably additive (in the sense that for any  sequence $A_n\downarrow \emptyset$ of Borel sets with bounded Lebesgue measure, $L^{(i)}(A_n)\to 0$ in probability as $n\to \infty$, see \cite{BNBV2011}), we can deduce that 
 $X_{0,t}^{(i)} \to 0$ 
in  probability as $t\to \infty$.
\end{proof}

\section{Details regarding the simulation study}\label{sim}
In Section \ref{SimStudy}, we simulate from a bivariate negative binomial trawl process, where both components have an exponential trawl function and their joint law is given by the bivariate negative binomial distribution as described in Example \ref{ex_dep}. In the simulation of the trawl process, we work with the compound-Poisson-type representation \eqref{CPP} and specify the jump size distribution as the bivariate logarithmic series distribution (BLSD) as in Example \ref{Ex}.

\subsection{Simulating from the bivariate logarithmic series distribution }
First of all, we describe how we can generate random samples ${\bf C}=(C_1,C_2)^{\top}$ from the BLSD with parameters $p_1,p_2$. The algorithm is based on the idea that we can simulate $C_1$ from the modified logarithmic series distribution (ModLSD) (with parameters $\tilde p_1=p_1/(1-p_2)$ and $\delta_1=\log(1-p_2)/\log(1-p_1-p_2)$) in a first step, and then $C_2$ can be simulated from the conditional distribution, given $C_1$, see e.g.~\cite{KempLoukas1978}. We note here, that if $C_1\equiv 0$, then $C_2|C_1$ follows the logarithmic distribution (with parameter $p_2$), and when $C_1>0$, then $C_2|C_1$ follows the negative binomial distribution with parameters $C_1$ and $p_2$, see  e.g.~\cite{KK1990}. 
We describe the simulation algorithm for the BLSD using pseudo code tailored to the ${\tt R}$ language. Throughout the section we use the abbreviation rv for random variable.

\begin{algorithm}
[Simulation from the bivariate logarithmic series distribution]\label{BLSD_SimAlgo}
\begin{algorithmic}[1]  \Statex
\State {\tt library(VGAM)} \Comment{Load the VGAM package in R.}
\Statex 

\Function{Sim-BLSD}{$N,p_1,p_2$}
 \State $\tilde p_1 \gets p_1/(1-p_2)$ \Comment{Calculate the parameters of the modified LSD.}
\State $\delta_1 \gets \log(1-p_2)/\log(1-p_1-p_2)$ 
\State   $L \gets {\tt rlog}(N,p_1)$ \Comment{Simulate $N$ i.i.d.~Log($p_1$) rvs.} 
\State $B \gets {\tt rbinom}(N,1,1-\delta_1)$ \Comment{Simulate $N$ i.i.d.~Bernoulli($1-\delta_1$) rvs.} 
\State $C_1 \gets L * B$ \Comment{Generate $N$ i.i.d.~ModLog($\tilde p_1, \delta_1$) rvs.}
\State $C_2 \gets {\tt numeric}(N)$
\For{$i$ in $1:N$}
$c_1 \gets C_1[i]$
 \If{$c_1==0$}
\State $C_2[i] \gets {\tt rlog}(1,p_2)$ \Comment{Simulate a Log($p_2$) rv.}\EndIf
 \If{$c_1>0$}
\State $C_2[i] \gets {\tt rnbinom}(1,size=c_1,prob=1-p_2)$ 
\Comment{Simulate a NB($c_1,p_2$) rv.}
\EndIf
 \EndFor
\State $C \gets {\tt cbind}(C_1,C_2)$ \Comment{Combine the component vectors to an $N\times 2$ matrix.}
\State \textbf{return} $C$
\EndFunction
\end{algorithmic}
\end{algorithm}

\subsection{Simulating the bivariate trawl process}\label{SimTrawlAppendix}
Next, we provide the pseudo code tailored to the ${\tt R}$ language which has been used to simulate the bivariate trawl process with exponential trawl function and bivariate negative binomial law (as in Example \ref{ex_dep}). Here we are using the same notation as in the general description of  Algorithm \ref{GenSimAlgo}. In addition,  
we denote by $bi$ the length of the burn-in period. I.e.~we will simulate the process over the time interval $[0,t]$ for $t=T+bi$ and then remove the initial burnin period, i.e.~we return the paths over the interval $(bi,bi+T]$.    

\begin{algorithm}
[Simulation from the bivariate trawl process]\label{Trawl_SimAlgo}
\begin{algorithmic}[1] 
\Statex
\State {\tt library(VGAM)} \Comment{Load the VGAM {\tt R} package and the function {\tt Sim-BLSD} defined above.} 
 \Function{Expfct}{($x,\lambda$)} \Comment{Choose an exponential trawl function.}
\State \textbf{return} $\exp(\lambda * x)$
\EndFunction
\Statex
\Procedure{Sim-Trawl}{$\Delta,T,bi,\lambda_1, \lambda_2, \alpha_1, \alpha_2, \kappa$}
\State 
$v \gets \kappa *\log(1+\alpha_1+\alpha_2)$ \Comment{Intensity of the driving Poisson process.}
\State
$p_1 \gets \alpha_1/(\alpha_1+\alpha_2+1);  
p_2 \gets \alpha_2/(\alpha_1+\alpha_2+1)$ \Comment{Parameters in the BLSD}
\State $N_t \gets {\tt rpois}(1,v*t)$ \Comment{Draw the number of jumps in $[0,t]$ from Pois($vt$).}
  \State  
  $\tau \gets {\tt sort}({\tt runif}(N_t, \min=0, \max=t))$ \Comment{ simulate the $N_t$ jump times from the ordered uniform distribution on $[0,t]$.}
 \State 
  $h \gets {\tt runif}(N_t, \min=0,\max=1)$ \Comment{Simulate the $N_t$ jump heights of the abstract spatial parameter of the Poisson basis from the uniform distribution on $[0,1]$.}
  \State 
  $m \gets {\tt Sim-BLSD}(N_t,p_1,p_2)$  \Comment{Draw the jump marks from the BLSD}
  \State  
  $C_1 \gets m[,1]; C_2 \gets m[,2]$ \Comment{Assign the jump marks to $C_1$ and $C_2$.}
  
  \Statex
   \State \Comment{Determine the number of  jumps  up to each grid point $k \Delta$ and
  store them in the  vector $V$.}
  \State $V \gets {\tt vector(mode = "numeric", length = floor}(t/\Delta))$ 
  \State $c \gets{\tt table(cut(jumptimes,seq}(0,t,1),{\tt include.lowest = TRUE))}$
  \State
 $ V[1]<- {\tt as.integer}(c[1])$
  \For{$k$ in $2:{\tt floor}(t/\Delta)$}
  \State
    $V[k] \gets V[k-1]+{\tt as.integer}(c[k])$
    \EndFor
  
  \Statex
  \For{$i$ in $1:2$}
  \Comment{Simulate the $i$th trawl process}
\State $TP_i \gets {\tt vector(mode = "numeric", length = floor}(t/\Delta))$
  \For{$k$ in $1:{\tt floor}(t/\Delta)$}
  
    \State $N_{k\Delta} \gets V[k]$  \Comment{Number of jumps  until time $k\Delta$.}
     \If{$N_{k\Delta}>0$}
      \State $d \gets k*\Delta - \tau[1:N_{k\Delta}]$ \Comment{Compute the time differences between $k\Delta$ and each jump time up to $k\Delta$.}
       \State $cond_i \gets 1-{\tt ceiling}(h[1:N_{k\Delta}]-{\tt Expfct}(-d,\lambda_i))$ \Comment{Check which points are in the trawl.}
      \State $TP_i[k] \gets {\tt sum}(cond_i*C_i[1:N_{k\Delta}])$ \Comment{Sum up the marks  in the trawl.}
      
     \EndIf
   
  \EndFor
  
\EndFor
 
\Statex 
\State $b_1 \gets bi/\Delta, b_2 = bi/\Delta+T/\Delta$ 
 \For{$i$ in $1:2$}
 \State   
  $TrawlProcess_i \gets TP_i[(b_1+1):b_2]$  \Comment{Cut off burn-in period.}
  \EndFor
\EndProcedure
\end{algorithmic}
\end{algorithm}

\end{appendix}

\section*{Acknowledgement} 
A.~E.~D.~Veraart acknowledges financial support by 
a Marie Curie
FP7 Career Integration Grant within the 7th European Union Framework Programme.

\bibliographystyle{agsm}
\bibliography{MultTrawlBib}

\begin{appendix}

\end{appendix}
\end{document}